\documentclass[11pt]{article}

\usepackage[margin=1in]{geometry}
\usepackage{amsmath, amssymb, amsthm}
\usepackage{mathtools}
\usepackage{mathptmx} 
\usepackage{booktabs}
\usepackage{array}
\usepackage{natbib}
\usepackage[colorlinks=true,allcolors=black]{hyperref}
\usepackage{setspace}
\usepackage{tikz}
\usetikzlibrary{positioning, arrows.meta, calc}
\usepackage{graphicx}
\usepackage{subcaption}
\usepackage{enumitem}
\usepackage{float}
\newfloat{algorithm}{ht}{loa}
\floatname{algorithm}{Algorithm}
\usepackage{comment}
\usepackage{xcolor}   
\usepackage[normalem]{ulem}   
\makeatletter
\DeclareFontFamily{OMS}{rsfs}{\skewchar\font=48}
\DeclareFontShape{OMS}{rsfs}{m}{n}{<-> zptmcm7y}{}
\makeatother

\newcommand{\E}{\mathbb{E}}
\newcommand{\Prb}{\mathrm{Pr}}
\newcommand{\indep}{\perp\!\!\!\perp}
\DeclareMathOperator*{\Var}{Var}
\DeclareMathOperator*{\Cov}{Cov}

\theoremstyle{plain}
\newtheorem{lemma}{Lemma}

\theoremstyle{remark}

\title{Algorithm or Creative? Disentangling Algorithmic Divergent Delivery from User Ad Preferences in Digital Platforms\thanks{We thank Dan Wang and participants at the Junior Faculty Meet at RRMB, Cornell Tech, AMA 2025, and ISMS Marketing Science 2025 for feedback on earlier versions of this paper.}}
\author{Pallavi Pal\thanks{School of Business, Stevens Institute of Technology. Email: \texttt{ppal2@stevens.edu}.} \and Anjana Susarla\thanks{Eli Broad College of Business, Michigan State University. Email: \texttt{asusarla@broad.msu.edu}.}}
\date{\today}

\begin{document}

\maketitle

\begin{abstract}
\noindent Online advertising platforms host hundreds of thousands of A/B tests, but the platform's delivery algorithm routes each creative to the audience it predicts will engage. Every two-arm test therefore conflates the creative's effect with the algorithm's targeting response, and adjusting for the realized audience is biased because audience is a post-treatment mediator. We propose a three-arm design that adds an arm exposing the algorithm to the treatment metadata while holding the user-facing creative identical to control, separately point-identifying the algorithmic delivery and ad creative effects. In a live Meta campaign with a women-targeted text fragment, the algorithmic delivery channel raises female impression share by $+2.27$ ppt, while the visible creative leaves delivery composition statistically unchanged. A conventional two-arm test would attribute the entire delivery shift to the creative. The design isolates the contribution of platform's algorithm to the outcome which is separable from creative content.


\end{abstract}

\newpage
\section{Introduction}
\label{sec:intro}

Digital platforms once functioned as a passive ``funnel'' through which advertisers could reach a chosen audience. 
Under this model, ``performance'' was viewed almost exclusively as an intrinsic property of the ad creative itself. This no longer fits modern platforms; an algorithm now actively filters user types for each ad creative. The emergence of AI and machine learning has turned this filtering process into a ``bigger black box.'' Consequently, the utility derived from a campaign depends jointly on two distinct channels: \textit{the Ad Creative Effect (ACE) and the Algorithm Filtering Effect (AFE).}

Isolating a clean causal effect of a creative in the presence of platform-mediated delivery is a known measurement problem: \citet{johnson2017ghost} address a related version --- the counterfactual-audience problem --- with the Ghost Ads design, which recovers the effect of exposure without altering the delivery algorithm. Our paper solves a different but complementary problem: separating the algorithm's contribution from the creative's. We introduce a new experimental design that empirically identifies these two effects through a general three-arm experimental design. 
To our knowledge, this is the first paper to provide an econometric method that causally identifies the effect of AI algorithm filtering separately from the creative within a single campaign --- a first step in understanding the AI's black-box output. We validate the methodology in two ways: a synthetic data generator that simulates the bidding and ad matching system of online ad platforms, and a field experiment on a live Facebook campaign. In both datasets we find a statistically significant effect of the algorithm filtering.

This filtering is not hypothetical. \citet{lambrecht2019algorithmic} systematically documented on Facebook that a gender-neutral STEM career ad is delivered disproportionately to men --- a canonical instance of how the algorithm's filtering shapes exposure even when the advertiser's intent is neutral. (Our own campaign reproduces both signatures: a stable female impression deficit and a female CPM premium that widens with visibility; Appendix~\ref{app:further-empirical}.) A subsequent literature has extended this to broader protected categories and platforms \allowbreak{\citep{ali2019discrimination, burtch2025characterizing}} and termed the phenomenon \emph{divergent delivery}.\footnote{{ On} Meta's platform alone, \citet{burtch2025characterizing} document roughly $182{,}000$ active A/B tests within their analysis window.} Because it operates silently inside the platform, every two-arm A/B test returns a number that blends the two effects. In the short run, the \textit{joint effect} of both the mechanisms provided by conventional two-arm A/B tests remains a useful metric. For a manager focused solely on immediate click-through rates or conversions, this single, blended number reflects the bottom-line reality of their current campaign performance. However, in the longer run, disentangling these effects can be useful for strategic long term creative decisions, regulatory policies, and platform's future direction for the algorithm: advertisers can correctly attribute performance to creative versus targeting; platforms can audit and tune their delivery algorithms; and regulators can scrutinize algorithmic exposure in protected categories. 



We show that the advertiser does not need privileged access to the platform's internals to break this blend; the platform's own A/B-testing tool is enough. Our design adds one arm in which the algorithm receives the treatment signal in the campaign metadata while the user-facing creative is held identical to the control. The contrast between this arm and the standard arms isolates the algorithm's filtering response and leaves the residual as the clean causal effect of the creative. Formally, we prove that this estimator is unbiased for the two effects and estimates the causal impact.  We model the algorithm as a mediator on the path from creative to outcome, and under three transparent assumptions, the method point-identifies the Algorithm Filtering Effect (AFE) and the Ad Creative Effect (ACE).\footnote{The three assumptions are  random arm assignment, excludability, and a match between the realized mediator distributions in the new arm and the treatment arms, verifiable in the campaign's initial window, before engagement feedback can alter delivery. Furthermore, we can identify these effects separately  without the cross-world sequential-ignorability assumption that observational mediation analysis requires.}

We also formalize a simulation method to create synthetic data using a game-theoretic model that replicates the auction bidding systems on these ad platforms; the model incorporates both the matching and the bidding system. We introduce noise in the way the actual matching quality is read by the ad platform, through noise introduced in the quality score from the click rate; we assume this noise is higher for female users. We set up an experiment and estimate the resulting divergent delivery. We find that noise in the quality score leads to delivery divergence picked up by our method; because it is a simulation, we can identify that the ad platform's score noise leads to delivery divergence originating in the match stage, where the ad platform chooses relevant ads for each user.

We then apply the method to a live Meta traffic-optimization campaign in which the treatment is the presence of a women-targeted text fragment. In the field experiment the algorithmic channel raises the female impression share by $+2.27$ percentage points while the visible creative leaves delivery composition statistically unchanged --- and a disaggregated decomposition shows the algorithm tilting delivery toward women most strongly within the 35--44 band, while the 65+ band's share of total impressions falls through a composition shift. The data shows that the delivery divergence on one demographic variable bleeds into other variables as well.

We think this framework will be useful in connecting the measurement problem to the algorithmic-fairness debate: the algorithm effect we identify is, by construction, the platform's contribution to differential exposure, separable from any choice the advertiser makes about creative content. The design generalizes to any delivery outcome (demographic share, geographic share, conversion mix) on any platform whose targeting metadata is separable from rendered content. We provide what is to our knowledge the first causal --- not descriptive --- decomposition of divergent delivery, in the sense that the delivery channel is point-identified separately from the creative channel within a single campaign.

\section{Literature review}
\label{sec:lit}

Our paper sits at the intersection of four literatures: research on algorithm matching and recommendations on platforms, the divergent-delivery literature in advertising that documents the measurement problem, the causal-mediation tradition that supplies our identification strategy, and the algorithmic-fairness literature that motivates the application.

As platforms have scaled, they have shifted from passive conduits into two-sided markets in which a machine-learning algorithm mediates the match between the two sides of the market, and a substantial literature shows that such curation systematically reshapes which products and users meet, and therefore the outcomes both sides realize \citep{fleder2009blockbuster, hosanagar2014village, peukert2024editor}. The closest work in this stream treats the algorithm as a determinant of platform outcomes; we extend it by treating the algorithm as a causal mediator on the path from advertiser content to user response, and designing an experiment that separates the algorithm's effect from the creative's.

In advertising, the visible symptom of algorithmic intermediation is \emph{divergent delivery}: ostensibly identical campaigns reach systematically different audiences even when targeting, budget, and placement are held fixed \citep{lambrecht2019algorithmic, ali2019discrimination, burtch2025characterizing}. The field's central measurement instruments are the platform's randomized holdout (lift) test \citep{gordon2019comparison} and, most commonly, the creative A/B test. \citet{braun2025where} document the A/B test's central weakness: the delivery algorithm serves each ad to a distinct, optimized mix of users even during the test, which can distort the magnitude and even the sign of the estimated contrast. They diagnose the problem but do not separate the two effects. Our paper suggests a method that proposes a decomposition --- to identify the divergence. The dominant methodological response to platform-mediated delivery in the single-creative case is \citet{johnson2017ghost}'s Ghost Ads design, which recovers the causal effect of ad exposure by tracking users who would have been treated in the intervention arm but were assigned to control --- a properly matched counterfactual for the algorithm-selected audience.  Where Ghost Ads asks ``what is the causal effect of my ad, given the audience the algorithm chose,'' we ask ``when two ads are compared, how much of the outcome difference is the algorithm's audience choice versus the ads themselves.'' \allowbreak{\citet{burtch2025characterizing} measure divergent delivery at scale and propose configuration-side mitigation; we causally identify the algorithm-determined share of that disparity.} Experimental studies establish that the skew is causal, \citet{ali2019discrimination} randomize creatives while holding targeting, budget, and bid fixed, and \citet{imana2024apparent} pair ads to separate qualification-driven from platform-driven skew --- but in these designs the visible creative and the metadata change together, so the algorithm's response cannot be separated from the creative's own effect. Our paper fills that gap by proposing an additional arm that helps isolate the effect of the delivery divergence.

The closest methodological literature is on experimental identification of mediation in the presence of a post-treatment mediator. \citet{imai2013experimental} formalize parallel and encouragement designs that randomize the mediator independently of the treatment, avoiding the cross-world ignorability assumption required by observational mediation analysis. \citet{heckman2015econometric} emphasize that the credibility of any decomposition rests entirely on the source of variation in the mediator, which is what an experimental design supplies. A related methodological hazard in platform experiments is interference: \citet{holtz2024interference} document via an Airbnb meta-experiment that more than $20\%$ of an individual-randomized estimate is interference bias, recoverable through cluster randomization. We adapt the parallel-design logic of \citet{imai2013experimental} to a setting where the mediator is a platform delivery algorithm rather than a user medium, and where the experimenter cannot set the algorithm's state directly but can manipulate the metadata signal the algorithm conditions on.

Because divergent delivery falls along demographic lines, measuring it is also a fairness question: in protected categories such as housing, credit, and employment, the algorithmic channel carries regulatory exposure that the creative does not. A practical obstacle is that advertisers seldom observe protected-class membership directly; \citet{kallus2022assessing} show how auxiliary data (e.g., the U.S.\ Census) and proxies such as surname and geolocation can bound the true disparity even when the protected attribute is unobserved.  The longer-run, strategic consequences of mismeasured fairness we defer to the policy discussion of Section~\ref{sec:policy}. Our paper sits closest to this strand: we supply the causal estimate of the algorithmic channel that fairness audits currently approximate with descriptive imbalance.

Across these strands, no prior work has integrated the platform-mediation mechanism, the divergent-delivery measurement problem, the causal-mediation toolkit, and the fairness stakes into a single causal estimate. This paper builds that bridge.

\section{Methodology}
\label{sec:methodology}

\subsection{Setup and the identification problem}

For each ad, an advertiser submits a creative. The platform's delivery algorithm uses the metadata to decide \emph{which} users see the ad; those matched users see the creative content. Given a treatment $D$, an outcome $Y$, and a third variable $S$, the structural relationships are summarized by Figure~\ref{fig:framework}.

\begin{figure}[h]
\centering
\begin{tikzpicture}[>=Latex, thick,
  box/.style={draw, rounded corners, inner sep=6pt, font=\small, align=center},
  mech/.style={draw, rounded corners, inner sep=6pt, font=\small, align=center, fill=blue!7, draw=blue!60!black}]
  \node[box, text width=4.8cm] (adv) at (0, 2.4) {\textbf{Advertiser}\\[2pt] submits creative};
  \draw[dashed, rounded corners=8pt, gray!70, thick] (-5, 0.7) rectangle (5, -2.4);
  \node[font=\footnotesize\bfseries, text=gray, fill=white, inner sep=3pt, anchor=west] at (-4.7, 0.7) {PLATFORM};
  \node[mech, text width=3.4cm] (a) at (-2.6, -0.85)
    {\textbf{Ad content}\\{\footnotesize(treatment $D$)}};
  \node[mech, text width=3.4cm] (f) at (2.6, -0.85)
    {\textbf{Delivery filtering}\\{\footnotesize(mediator $S$)}};
  \draw[->] (a.east) -- node[above, font=\footnotesize\itshape] {via metadata} (f.west);
  \node[box, text width=3.4cm] (y) at (0, -4.3) {\textbf{Engagement $Y$}\\{\footnotesize click, CTR}};
  \draw[->] (adv.south) -- (a.north);
  \draw[->] (a.south) to[bend right=25] node[left=2pt, font=\footnotesize] {{creative (ACE)}} (y.west);
  \draw[->] (f.south) to[bend left=25] node[right=2pt, font=\footnotesize] {{algorithm (AFE)}} (y.east);
\end{tikzpicture}
\caption{Advertiser-user interaction on the ad platform}
\label{fig:framework}
\end{figure}

We compare two treatments, $A$ and $B$, with respect to an outcome $Y$. Although treatment is assigned at random, each subject's exposure to treatment is shaped by a delivery mechanism $S$ --- in our application, an ad delivery algorithm that determines which advertisements are served to which users. Additionally, as $S$ is itself a post-treatment variable that responds to treatment and in turn influences the outcome, it satisfies the formal definition of a mediator. More details are given in Appendix~\ref{app:methodology-details}.

Let $Y_i(a, s)$ denote the outcome of unit $i$ under treatment $a \in \{A, B\}$ and $S_i(a)$ denote the mediator value for unit $i$ realizes under $a$. The observed outcome under treatment $a$ is $Y_i(a, S_i(a))$.

A naive comparison of treatment arms estimates
\begin{equation}
\E\!\bigl[Y(A, S(A))\bigr] \;-\; \E\!\bigl[Y(B, S(B))\bigr],
\label{eq:naive}
\end{equation}
which is the \emph{total effect} of $A$ versus $B$. Equation~\eqref{eq:naive} is a well-defined causal contrast under randomization of treatment alone, but it is not the contrast a researcher typically wants when $S$ is itself altered by the treatment: it confounds the \emph{direct} effect of receiving $A$ rather than $B$ with the \emph{indirect} effect operating through changes in the mediator. The treatment shifts both the direct channel and the mediator; the total effect mixes the two channels in an unidentifiable proportion under a single-stage design. In the appendix we discuss why including covariates related to mediator S does not solve the problem.

\subsection{The proposed three-arm design}

We add a third arm in which subjects receive treatment $A$, but the mediator is intervened upon and set to a value drawn from the distribution that $S$ takes under $B$. Concretely, subjects are randomly assigned to one of three arms:
\begin{center}
\begin{tabular}{clll}
\toprule
Arm & Treatment & Mediator & Observed quantity \\
\midrule
$C$   & $A$ & $S(A)$ & $\E[Y(A, S(A))]$ \\
$T_1$ & $A$ & $S(B)$ (imposed) & $\E[Y(A, S(B))]$ \\
$T_2$ & $B$ & $S(B)$ & $\E[Y(B, S(B))]$ \\
\bottomrule
\end{tabular}
\end{center}

\noindent $T_1$ is the new ingredient: subjects assigned to it receive the substantive treatment $A$, but the mediator is fixed by the experimenter to a value drawn from the distribution of $S(B)$. This decoupling of treatment from mediator is the central design feature.

\subsection{Identification of AFE and ACE}
\label{sec:identification}

We identify two effects: the \emph{Algorithm Filtering Effect} (AFE), the effect of the algorithm's ad filtering on the outcome, and the \emph{Ad Creative Effect} (ACE), the effect of the ad creative that users see. Their sum is the total effect.

Given the three observed means above, define
\begin{align}
\mathrm{AFE} &\;=\; \E[Y(A, S(B))] - \E[Y(A, S(A))], \label{eq:afe}\\
\mathrm{ACE} &\;=\; \E[Y(B, S(B))] - \E[Y(A, S(B))], \label{eq:ace}\\
\mathrm{TE}  &\;=\; \E[Y(B, S(B))] - \E[Y(A, S(A))]. \label{eq:te}
\end{align}

AFE isolates the algorithm filtering (treatment held at $A$, $S$ shifts $S(A) \to S(B)$); ACE isolates the creative channel ($S$ held at $S(B)$, treatment shifts $A \to B$).

The decomposition
\begin{equation}
\mathrm{TE} \;=\; \mathrm{AFE} + \mathrm{ACE}
\label{eq:decomp}
\end{equation}
follows algebraically from \eqref{eq:afe}--\eqref{eq:te}.

In the three-arm design, each of the three means in~\eqref{eq:afe}--\eqref{eq:te} is identified by a simple sample mean within an arm. No cross-equation regression and no sequential ignorability assumption is required. The only assumptions invoked are:
\begin{enumerate}
\item[(A1)] \emph{Random assignment to arm.} Arm indicator $Z \in \{1,2,3\}$ is independent of all potential outcomes.
\item[(A2)] \emph{Excludability.} The arm indicator $Z$ affects $Y$ only through the realized $(D, S)$ pair specified in the table.
\item[(A3)] \emph{Successful manipulation of $S$ in $T_1$.} Conditional on $Z=2$, the realized $S$ is drawn from the marginal distribution of $S(B)$, independently of the potential outcomes $\{Y_i(a,s)\}$.
\end{enumerate}
Assumption~(A1) is delivered by the randomization. Assumption~(A2) is the standard SUTVA-style excludability assumption. Assumption~(A3) is the only design-specific assumption and is verifiable in the campaign's initial window: before engagement feedback can alter the matching rule, the visible creative has no channel to delivery, so one can compare the empirical distribution of $S$ in $T_1$ to that of $S$ in $T_2$ and test their equality by a Kolmogorov--Smirnov or randomization-based test --- within this window, a divergence would uniquely indicate manipulation failure. The distributional component of (A3) is what the comparison of $T_1$ and $T_2$ verifies; the independence component holds by construction when the platform applies the same serving rule to two disjoint, randomly assigned audiences, since the rule's draw of who is matched cannot depend on which arm's audience it is operating on.\footnote{This assumes the matching rule is static within the experimental window. If the serving rule updates on engagement feedback during the campaign, the realized mediator distributions in $T_1$ and $T_2$ can diverge, which the distributional comparison above would detect.} Let $\bar{Y}_z$ be the mean outcome in arm $z$. We now state the identification result.

\begin{lemma}[Unbiasedness of the three-arm decomposition]
\label{lem:unbiased}
Under assumptions (A1)--(A3), the estimators
\[
\widehat{\mathrm{AFE}} = \bar{Y}_2 - \bar{Y}_1,
\qquad
\widehat{\mathrm{ACE}} = \bar{Y}_3 - \bar{Y}_2,
\qquad
\widehat{\mathrm{TE}}  = \bar{Y}_3 - \bar{Y}_1
\]
are unbiased for the corresponding population parameters $\mathrm{AFE}$, $\mathrm{ACE}$, and $\mathrm{TE}$ defined in equations~\eqref{eq:afe}--\eqref{eq:te}, and the sample-level decomposition $\widehat{\mathrm{TE}} = \widehat{\mathrm{AFE}} + \widehat{\mathrm{ACE}}$ holds exactly.\footnote{See Appendix~\ref{app:proof} for proof}
\end{lemma} 
The design's unbiasedness rests on two features. First, each effect is the difference of arm-means made comparable by physical randomization rather than statistical adjustment. Second, the enforced marginal $S(B)$ distribution in $T_1$ replaces $S(A)$, thereby isolating the algorithm effect ($S$) from the creative effect. The decomposition is therefore robust to unmeasured confounding between $S$ and $Y$, which is the principal source of bias in regression-based mediation.

\subsection{Estimation}

Let $Z_i \in \{1, 2, 3\}$ index the arm of unit $i$ and let $n_z = \sum_i \mathbf{1}\{Z_i = z\}$. The point estimators are simple group differences; valid standard errors come from the heteroskedasticity-consistent estimator
\begin{equation}
\widehat{\Var}(\bar{Y}_z - \bar{Y}_{z'}) \;=\; \frac{s_z^2}{n_z} + \frac{s_{z'}^2}{n_{z'}},
\label{eq:varhat}
\end{equation}
where $s_z^2$ is the within-arm sample variance. In a regression representation,
\begin{equation}
Y_i \;=\; \mu_1 \mathbf{1}\{Z_i = 1\} + \mu_2 \mathbf{1}\{Z_i = 2\} + \mu_3 \mathbf{1}\{Z_i = 3\} + e_i,
\label{eq:saturated}
\end{equation}
and contrasts of the $\hat{\mu}_z$ recover $\widehat{\mathrm{AFE}}$, $\widehat{\mathrm{ACE}}$ and $\widehat{\mathrm{TE}}$.\footnote{Pre-treatment covariates may be added via the interacted specification of \citet{lin2013agnostic} to improve precision without compromising unbiasedness.}

\bigskip

\section{Simulation evidence}
\label{sec:simulation}

In this section, we test our method on synthetic data generated from an auction model in which the three potential outcomes are known by construction. The simulation helps confirm numerically that the estimators of Lemma~\ref{lem:unbiased} recover the population AFE, ACE, and TE in a setting that closely resembles ad delivery. Second, it provides a theoretical reasoning of the bias, by mapping the empirical AFE magnitude onto the noise variance of the platform's predicted click-rate signal.\footnote{We assume the auction uses the bid alone. This simplification preserves the role of the algorithmic filtering channel, separable from any strategic bidding effects. The non-stylized version would include a quality-score-weighted bid. We chose this version so that the effect of algorithmic filtering can be cleanly presented.}

\subsection{Model setup}
\label{sec:sim-theory}

The simulator implements a two-stage auction: algorithmic matching by quality score, then a second-price auction on bids among matched advertisers.\footnote{Each advertiser has one ad, unless the advertiser is running an experiment.}

\paragraph{Advertisers \& users primitives:} The platform has $N_u$ users and $N_i$ advertisers. User $u$ has gender $g_u \in \{M, F\}$; advertiser $i$ has per-click value $v_i \sim F_v$ and ad type $a(i)$. Per-pair true click rates satisfy $r_{u,i} \sim \mathrm{Beta}(\alpha, \beta)$ with $\E[r] = \mu_r$. The Beta family produces the L-shape observed empirically in ad-CTR distributions (further discussion in Appendix~\ref{app:sim-details}).

 \paragraph{Stage 1 and platform's ad filtering:} In stage 1, the platform ranks advertisers for each user $u$ by the matching quality score $s_{u,i} = r_{u,i} + \varepsilon_{u,i}$, with independent noise $\varepsilon_{u,i} \sim \mathrm{N}(0, \sigma^2_{a(i), g_u})$, and retains the top $q$ fraction to enter the auction. Signal noise is heteroskedastic across (ad-type $\times$ gender) --- this is the sole channel through which the algorithm produces AFE. 
%

%
\paragraph{Stage 2 and the bidding strategy:} In stage 2, an auction is held for each user to decide the winning ad; the winner is determined using bids in a second-price sealed-bid auction (SBA). Under SBA, submitting the true per-click valuation is a weakly dominant strategy. Therefore, each bidder bids their valuations, $b_i = v_i$.

\paragraph{The ad experiment in the simulation}
Let's assume there is a focal advertiser who created an ad experiment to compare Ad A to Ad B. We assume that the bid is set at three top-tail percentiles of $F_v$,  in order to get substantial wins for the focal advertiser. We set up a two-ad experiment, Ad A and Ad B, through the three-arm design. Specifically,  the standard deviation of the matching-signal noise $\varepsilon$ is zero, except female users for Ad $B$. We assume that the other advertisers are split near equally across the two ad types, which mainly affect the quality noise across advertisers.\footnote{For more details, refer to Appendix~\ref{app:sim-details}.}

We assume control ad being Ad A with assigned score $s_{u,i}(A)$, we switch to treatment 1 (i.e., $T_1$), which is still ad A for users, meaning the click rate comes from Ad $A$ but ad filtering changes as  the platform-side signal changed from $s_{u,i}(A)$ to $s_{u,i}(B)$. This alters the signal-noise structure without disturbing the true click-rate primitive or the auction format. This arm captures AFE, which is exactly the impression-level expression of any matching-stage primitive that differs by creative. The second treatment (i.e., $T_2$) captures the complete Ad B setup, i.e., both the click rate and platform-side signal being ($B, s_{u,i}(B)$), which captures ACE.

\begin{center}
\begin{tabular}{cccl}
\toprule
Arm $z$ & $(Ad, S)$& Focal creative & Focal signal-noise $\sigma_{0, g_u}$ \\
\midrule
1 & $(A, S(A))$ & $A$ & $0$ for all $g_u$ \\
2 & $(A, S(B))$ & $A$ & $\sigma \cdot \mathbf{1}\{g_u = F\}$ \\
3 & $(B, S(B))$ & $B$ & $\sigma \cdot \mathbf{1}\{g_u = F\}$ \\
\bottomrule
\end{tabular}
\end{center}
Note only the focal's configuration shifts with the arm; non-focal advertisers retain their own incumbent $(a(i), \sigma)$ assignment, identical across arms. This is the direct simulation analogue of $T_1$ in the methodology: the creative is held at $A$ while the mediator is drawn from $S(B)$.

Under the no-creative-effect specification (the true click rate does not depend on which creative an advertiser runs), $T_1$ and $T_2$ impose the same joint law on $(r_{u, 0}, s_{u, 0})$ for the focal. Because no creative effect is built into the model, the estimator should detect none: any non-zero $\widehat{\mathrm{ACE}} = \bar{Y}_3 - \bar{Y}_2$ must come from finite-sample noise, not from a substantive direct effect. We use this as a check that the estimator does not identify a biased creative effect.\footnote{To accommodate a substantive creative effect, the click-rate distribution can be made creative-specific without changing the rest of the construction; we report the no-creative-effect specification as the cleanest base case.}

\paragraph{Simulation algorithm} 
We simulate $N_u = 10{,}000$ users per round, each with a gender $g_u \in \{M, F\}$ drawn at the U.S.\ Facebook share $\Prb(g_u = M) = 0.462$. There are $N_i = 500$ advertisers indexed $i = 0, \ldots, N-1$, of whom one (the focal advertiser, $i = 0$) corresponds to the experimenter and the remaining $N-1$ are incumbent competitors.\footnote{We report the results at $R = 1{,}000$ Monte-Carlo rounds throughout the paper.} So $\sigma^2 = 9 \times 10^{-4}$ --- about $1.8\times$ the variance of $r$ itself. For each (user, advertiser) pair we track two latent quantities, true click rate $r_{u,i}$ and matching quality score $s_{u,i}$.

The two-stage delivery mechanism mirrors a production ad-auction pipeline.\\
\emph{- Stage 1 (matching):} For each user $u$, the platform ranks all $N$ advertisers by their signal $s_{u,i}$ and retains the top $q = 0.10$ to participate in the auction. \\
\emph{- Stage 2 (auction):} among matched advertisers, the highest bid wins and pays the second-highest bid, and the winner's ad is shown to the user.\footnote{ The signal $s$ determines auction participation but does not enter the within-auction ranking, which is on bid alone. Calibration: $\mu_r = 0.02$, $\Var(r_{u,i}) = 5 \times 10^{-4}$, $F_v = \mathrm{LogNormal}(0, 0.5^2)$ }

\subsection{Simulation results}
\label{sec:sim-results}

\paragraph{The effect of quality score noise on matched distribution across gender}
Figure~\ref{fig:sim-A} reports average percentage of female matched in each arm. The focal's matched pool is $53.7\%$ female in $C$ and $69.9\%$ female in $T_1$. The gap is the result of quality-score noise, which we created to be higher in $T_1$ for the simulation. Arm $T_2$ sits at $69.9\%$ as well, since its quality-score distribution is identical to $T_1$'s.

\begin{figure}[h]
\centering
\caption{Female share of the focal advertiser's matched pool by arm}
\label{fig:sim-A}
\includegraphics[width=0.75\textwidth]{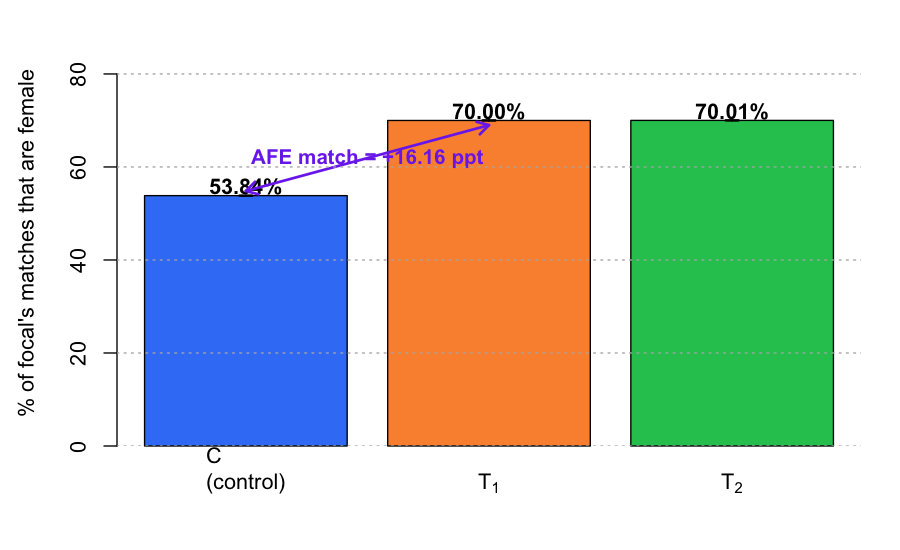}
\end{figure}

\paragraph{Two follow-up questions, both addressed in Appendix~\ref{app:bid-transfer}.}
Figure~\ref{fig:sim-A} pins down the algorithm bias at the matching stage. Two natural follow-ups are (i) how does that bias transfer to what the advertiser actually observes at the impression stage as the focal bid varies, and (ii) does the three-arm design distort the platform's overall serve rate. Appendix~\ref{app:bid-transfer} defines the match-stage and impression-stage measures and theoretically shows the correlation with bid level and how much of match-stage bias will be transferred to impression-level bias.

The outcome is $Y \equiv W_F / W$, the female share of focal impressions per round. We average the paired round-level differences (Lemma~\ref{lem:unbiased}) over $R = 1{,}000$ rounds and report 95\% CIs based on $t_{R-1,\,0.975}$. Table~\ref{tab:sim-decomp} reports the decomposition at $\sigma = 0.030$ with the focal bid at the 99.5th percentile of $F_v$.

\begin{table}[h]
\centering
\caption{Simulated decomposition on female delivery share}
\label{tab:sim-decomp}
\begin{tabular}{lccc}
\toprule
Name & Estimate (pp) & SE (pp) & 95\% CI (pp) \\
\midrule
$\widehat{\mathrm{AFE}} = \bar{Y}_2 - \bar{Y}_1$ & $+16.21$ & $0.03$ & $[+16.14,\, +16.27]$ \\
$\widehat{\mathrm{ACE}} = \bar{Y}_3 - \bar{Y}_2$ & $+0.00$  & $0.02$ & $[-0.04,\, +0.05]$ \\
$\widehat{\mathrm{TE}}  = \bar{Y}_3 - \bar{Y}_1$ & $+16.21$ & $0.03$ & $[+16.14,\, +16.28]$ \\
\bottomrule
\end{tabular}
\end{table}

The sharply positive $\widehat{\mathrm{AFE}}$ and near-zero $\widehat{\mathrm{ACE}}$ confirm Lemma~\ref{lem:unbiased} numerically.

\section{Empirical analysis}
\label{sec:empirical}

We run an ad-delivery experiment on Meta. The experiment takes a generic STEM job ad and runs it for one week. This category of ads --- STEM job ads --- is known to have gender inequality in ad delivery \citep{lambrecht2019algorithmic}. We create a treatment ad where we add a women-targeted text fragment, and create two versions, $T_1$, where the text ``Women'' is present only in the ad metadata (visible to the delivery algorithm, not the user). We are able to achieve this by adding text that is too long to be displayed in the ad window and thus is hidden in the background. Furthermore, in $T_2$, we include the text in both the ad metadata and the user-facing body, making it visible in the ad. Figure~\ref{fig:three-arms} shows the user-facing creative for each arm as it appeared to the user. 

\begin{figure}[h!]
\centering
\begin{subfigure}[t]{0.30\textwidth}
\centering
\includegraphics[width=\textwidth]{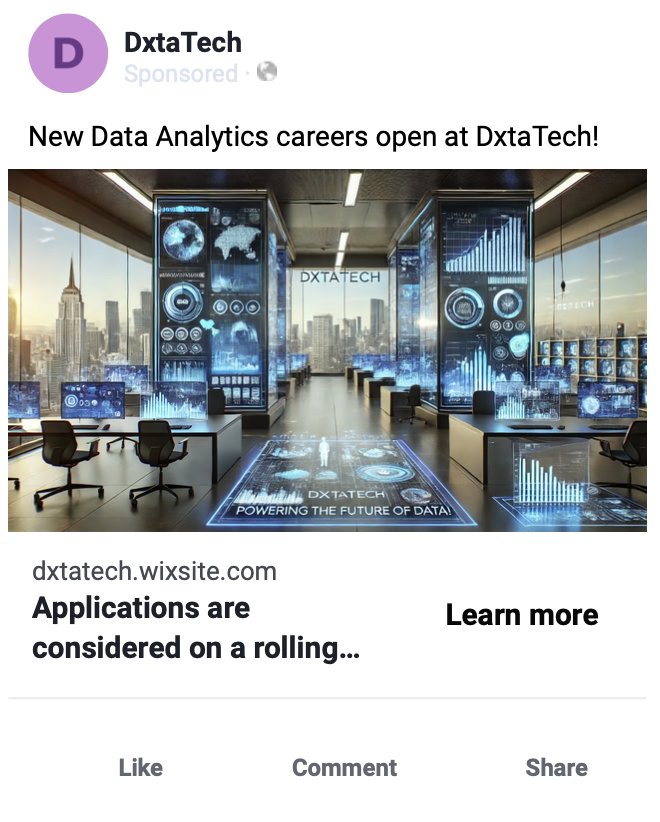}
\caption{$C$: control. Generic creative; no women-targeted text in metadata or body.}
\label{fig:arm1}
\end{subfigure}\hfill
\begin{subfigure}[t]{0.30\textwidth}
\centering
\includegraphics[width=\textwidth]{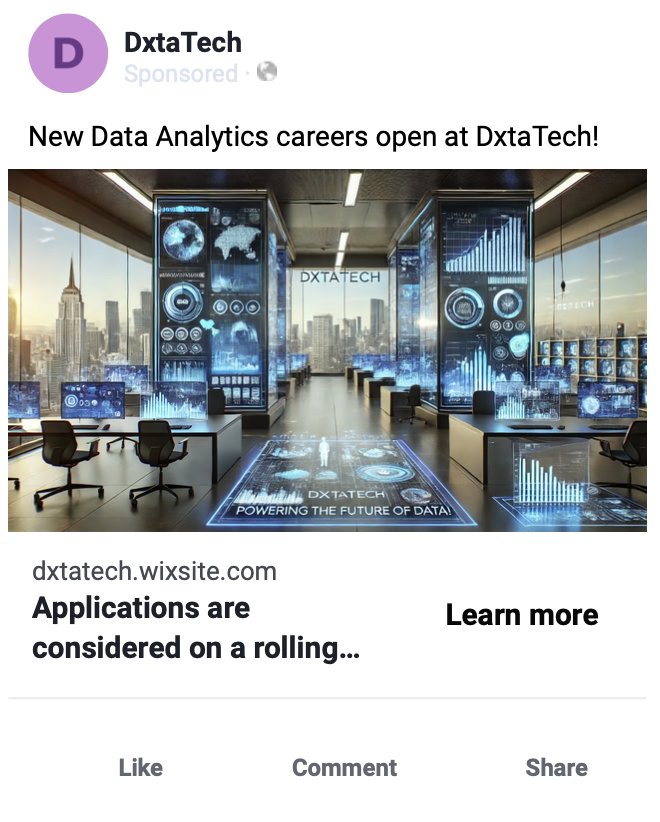}
\caption{$T_1$: metadata only. Body text is identical to $C$; women-targeted fragment appears only in the metadata visible to the algorithm.}
\label{fig:arm2}
\end{subfigure}\hfill
\begin{subfigure}[t]{0.30\textwidth}
\centering
\includegraphics[width=\textwidth]{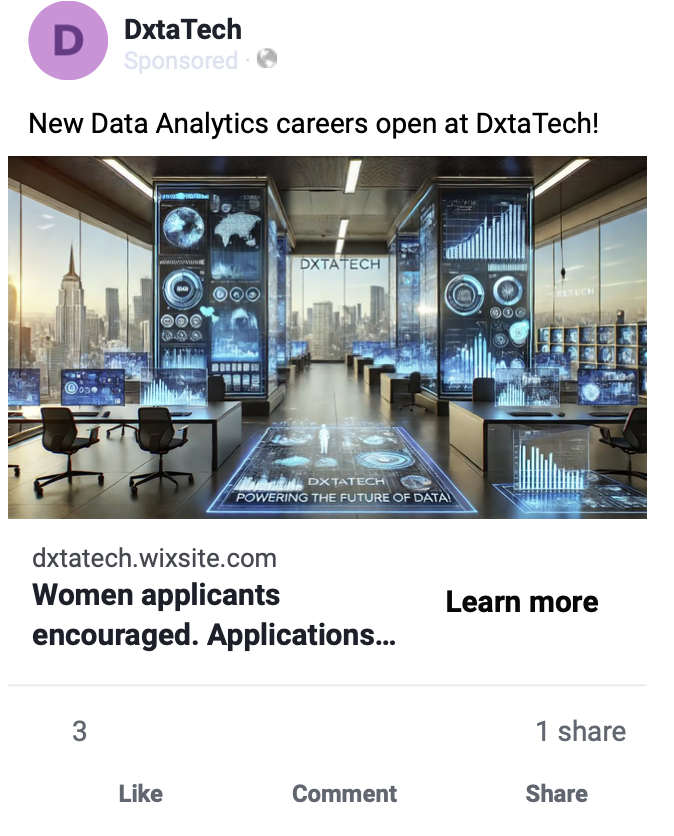}
\caption{$T_2$: visible. Women-targeted fragment in both metadata and user-facing body text.}
\label{fig:arm3}
\end{subfigure}
\caption{The three experimental arms as visible to users}
\label{fig:three-arms}
\end{figure}

\subsection{Data and design}

The data is recorded at the level of (ad creative $\times$ age band $\times$ gender $\times$ day $\times$ time-slot) cells, with impressions, link clicks, reach, and spend aggregated within each cell.
The three arms within each bid level correspond to the methodology's $(D, S)$ conditions of Section~\ref{sec:identification} as follows:

\begin{center}
\begin{tabular}{cll}
\toprule
Arm & Ad creative & $(D, S)$ condition \\
\midrule
$C$   & Generic creative, no ``wd F'' anywhere & $(A, S(A))$ \\
$T_1$ & Generic creative, ``wd F'' in metadata only & $(A, S(B))$ \\
$T_2$ & Visibly female-coded creative, ``wd F'' in both & $(B, S(B))$ \\
\bottomrule
\end{tabular}
\end{center}

The estimation is replicated at high and low bid levels; the high-bid experiment is the primary source of inference; the results from the low-bid experiment are discussed in Appendix~\ref{app:further-empirical}.

\paragraph{Features intentionally kept constant and varied across arms:} The ad creative differs across the three arms (the manipulation of interest), and the bid amount differs across the two parallel three-arm experiments (a secondary source of heterogeneity). We keep the campaign objective as traffic, and ad placement as a single Facebook news-feed surface. The ad format, budget, target audience, schedule and duration are held identical across the three arms within each bid level. The ad creative and the bid amount are the two variables intentionally varied.\footnote{More details on the setup in Appendix~\ref{app:campaign-config}}

\subsection{Data description}
\label{sec:descriptives}

We restrict the tracking to the campaign's initial week so that learning effects do not contaminate the static delivery decision; this depresses absolute click counts, so impression share is the primary outcome and CTR is reported as a secondary outcome. 

Table~\ref{tab:armlevel} reports arm-level totals at each bid. Volume is balanced across arms at the high bid (33.3--33.6k impressions, $\approx$33k unique users, CPM \$4.08--\$4.12); the low-bid arms are small due to a lower probability of a win.

\begin{table}[ht]
\centering
\small
\caption{Arm-level descriptives.}
\label{tab:armlevel}
\begin{tabular}{llrrrrrr}
\toprule
Bid & Arm & Impressions & Reach & Clicks & CTR (\%) & CPM (\$) & Imp/User \\
\midrule
High & $C$   & 33{,}325 & 33{,}042 & 23 & 0.069 & 4.075 & 1.009 \\
High & $T_1$ & 33{,}277 & 32{,}971 & 32 & 0.096 & 4.115 & 1.009 \\
High & $T_2$ & 33{,}554 & 33{,}337 & 29 & 0.086 & 4.080 & 1.007 \\
\midrule
Low  & $C$   &      485 &      485 &  0 & 0.000 & 1.876 & 1.000 \\
Low  & $T_1$ &      421 &      420 &  0 & 0.000 & 1.956 & 1.002 \\
Low  & $T_2$ &      698 &      695 &  0 & 0.000 & 1.945 & 1.004 \\
\bottomrule
\end{tabular}
\end{table}

Moving from $C$ to $T_1$ raises the female share of impressions (Table~\ref{tab:delivery-high}) from 34.8\% to 36.9\% --- a $+2.07$ pp shift produced by the metadata-only signal, with delivery volume held essentially constant. The reallocation is not uniform across age; Section~\ref{sec:empirical-age-gender} quantifies this heterogeneity. The visible-content arm ($T_2$) is where engagement responds: female $\Delta$CTR widens from $0.000$ pp in $C$ to $+0.084$ pp in $T_2$. In the initial analysis it seems that the algorithm's response to the ``wd F'' metadata is to reallocate working-age impressions toward women.

\begin{figure}[h!]
\centering
\includegraphics[width=0.78\textwidth]{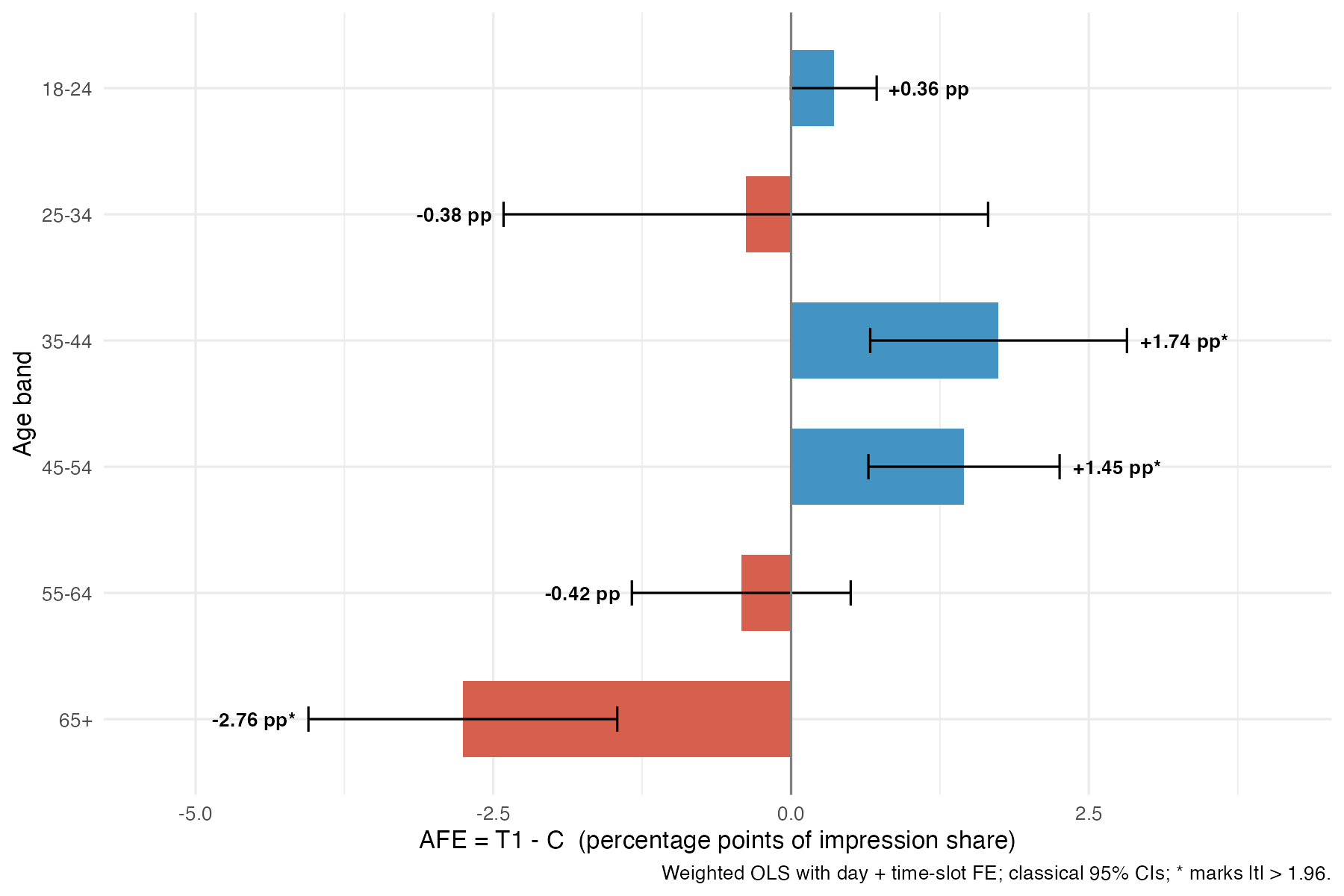}
\caption{$T_1 - C$ shift in age-band impression share, gender pooled (high bid). 95\% CIs; $*$ marks $|t|>1.96$.}
\label{fig:age-AFE}
\end{figure}

\subsection{Empirical strategy}
\label{sec:inference}

We begin by estimating the two effects through the following regression:
\begin{equation}
y_{z,d,t,a,g} \;=\; \alpha + \beta_1 \mathbf{1}\{z = T_1\} + \beta_2 \mathbf{1}\{z = T_2\} + \lambda_a + \gamma_d + \delta_t + \varepsilon_{z,d,t,a,g},
\label{eq:main-spec}
\end{equation}
where the outcome $y$ is an indicator that the impressions were delivered to women, for each aggregate observation at the arm-day-time-slot-age-gender level, weighted by aggregated impressions; and $\lambda_a$, $\gamma_d$, $\delta_t$ are age, day, and time-slot fixed effects.\footnote{Standard errors are clustered at the (day $\times$ time-slot) level (CR2, $G = 19$),
the level at which delivery-auction shocks are shared. The method's minimum detectable effect (80\% power, 5\% two-sided) is $\approx 2.2$~pp for the delivery outcome --- the observed AFE of $2.27$~pp sits just above this threshold --- and $\approx 0.07$~pp for CTR against a $0.08\%$ base, which is why the engagement decomposition in Appendix~\ref{app:ctr} is reported as descriptive rather than a powered test.} The time and day controls absorb common delivery shocks; for instance, afternoon delivery has a female share roughly $8$ percentage points higher than morning in the raw data (about $3.5$~pp conditional on the other controls), and the age controls hold the age composition fixed.

Contrasts recovered from the single regression~\eqref{eq:main-spec} are $\widehat{\mathrm{AFE}} = \hat{\beta}_1$, $\widehat{\mathrm{TE}} = \hat{\beta}_2$, and $\widehat{\mathrm{ACE}} = \hat{\beta}_2 - \hat{\beta}_1$.\footnote{We use the covariance-corrected variance $\Var(\hat{\beta}_1) + \Var(\hat{\beta}_2) - 2\Cov(\hat{\beta}_1, \hat{\beta}_2)$ for the difference. Appendix~\ref{app:wildboot} reports wild cluster bootstrap Romano-Wolf adjusted $p$-values on these three contrasts, with Rademacher weights drawn at the (day $\times$ time-slot) cluster level, matching the clustering of the analytic standard errors.}

\subsection{Results}
\label{sec:empirical-delivery}
Table~\ref{tab:delivery-high} reports the three-arm decomposition of female impression share.

\begin{table}[h]
\centering
\caption{Decomposition of female impression share, high-bid experiment.}
\label{tab:delivery-high}
\begin{tabular}{lrrr}
\toprule
Quantity & $C$ & $T_1$ & $T_2$ \\
\midrule
Female impressions & 11{,}595 & 12{,}267 & 12{,}141 \\
Total impressions  & 33{,}325 & 33{,}277 & 33{,}554 \\
Female share (\%)  & 34.79   & 36.86   & 36.18   \\
\midrule
$\widehat{\mathrm{AFE}}$ (algorithm channel) & & $+2.27$ \,(0.77) & \\
$\widehat{\mathrm{ACE}}$ (visible content)   & &                 & $-0.93$ \,(0.56) \\
$\widehat{\mathrm{TE}}$  (total)             & &                 & $+1.33$ \,(1.01) \\
\bottomrule
\end{tabular}
\par\smallskip
\footnotesize\emph{Note.} Standard errors in parentheses.
\end{table}

The algorithm channel tilts delivery toward women by $+2.27$ percentage points within age bands. The visible-content contrast is $-0.93$ pp (SE $0.56$), statistically indistinguishable from zero. This comparison also serves as the manipulation check described in Section~\ref{sec:identification}: at this early stage of the campaign, delivery in $T_1$ and $T_2$ should coincide, and neither the level contrast ($p = 0.095$) nor its per-age decomposition (joint $p = 0.332$, Section~\ref{sec:empirical-age-gender}) rejects that they do. For scale, the $+2.27$~pp tilt is a $6.5\%$ relative increase on the control arm's $34.8\%$ female share, closing roughly $7$--$8\%$ of the male--female delivery gap --- the expected order of magnitude for a metadata-only nudge holding creative, budget, and placement fixed. The practical conclusion is about attribution: a two-arm audit observes only the total effect ($+1.33$~pp) and would attribute the entire delivery shift to the creative; separating the channels shows the shift is algorithmic. This distinction matters for advertisers deciding whether to revise the creative, and for platforms accountable for the delivery outcome.

A robustness table reporting the same decomposition under alternative fixed-effects specifications is provided in Appendix~\ref{app:further-empirical}.

\subsection{Heterogeneity of the algorithm channel}
\label{sec:empirical-age-gender}

The aggregate decomposition in Section~\ref{sec:empirical-delivery} averages the within-age tilt across age bands. We test whether the two channels operate uniformly across viewer age by interacting the arm indicators with the age bands in equation~\eqref{eq:main-spec} --- the same panel, weights, and clustering as the main specification:
\[
y_{z,d,t,a,g} = \lambda_a + \beta_1 T_{1,z} + \beta_2 T_{2,z} + \sum_{a \neq a_0} \theta_{1,a}\, T_{1,z} \mathbf{1}\{a\} + \sum_{a \neq a_0} \theta_{2,a}\, T_{2,z} \mathbf{1}\{a\} + \gamma_d + \delta_t + \varepsilon.
\]
The per-age treatment effects trace how strongly the algorithm tilts the gender mix within each band.

Table~\ref{tab:age-decomp} reports the implied per-age effects; Figure~\ref{fig:within-age} plots them.

\begin{table}[h]
\centering
\caption{Per-age within-age treatment effects on the female impression share, high-bid experiment.}
\label{tab:age-decomp}
\begin{tabular}{lrrr}
\toprule
Age band & AFE (algorithm) & ACE (visible content) & TE (total) \\
\midrule
18--24 & $-0.84$ \,(2.60) & $+1.39$ \,(2.85) & $+0.55$ \,(2.49) \\
25--34 & $+2.81$ \,(2.71) & $-1.55$ \,(1.26) & $+1.26$ \,(2.69) \\
35--44 & $+3.86$ \,(1.96) & $-1.42$ \,(1.61) & $+2.44$ \,(2.56) \\
45--54 & $+1.38$ \,(1.08) & $-0.78$ \,(0.79) & $+0.60$ \,(1.02) \\
55--64 & $+2.89$ \,(0.84) & $-1.32$ \,(1.17) & $+1.56$ \,(0.98) \\
65+    & $+0.23$ \,(1.16) & $+0.37$ \,(1.41) & $+0.60$ \,(1.28) \\
\bottomrule
\end{tabular}
\par\smallskip
\footnotesize\emph{Note.} Effects in percentage points on the within-age female share. Standard errors in parentheses.
\end{table}

The algorithm's within-age tilt toward women is concentrated in the working-age segment, specifically the 35--44 and 55--64 bands. Jointly, the algorithm channel is heterogeneous across age bands: the test of $\theta_{1,a} = 0$ for all $a$ gives $p = 0.002$. This signifies that when algorithmic delivery diverges around one demographic variable, it is heterogeneously applied across other demographic variables, making it a complex system to disentangle. The creative channel's within-age response is not: the joint test of the interaction differences $\theta_{2,a} - \theta_{1,a} = 0$ (the ACE heterogeneity null, since $\mathrm{ACE} = \beta_2 - \beta_1$) gives $F(5, 18) = 1.24$, $p = 0.332$, and total-effect heterogeneity is likewise absent ($F(5, 18) = 0.96$, $p = 0.468$). The asymmetry is informative: the algorithm \emph{targets}, its response to the metadata signal depends on viewer age, while the creative channel shows no detectable within-age response.

\begin{figure}[ht]
\centering
\includegraphics[width=0.85\textwidth]{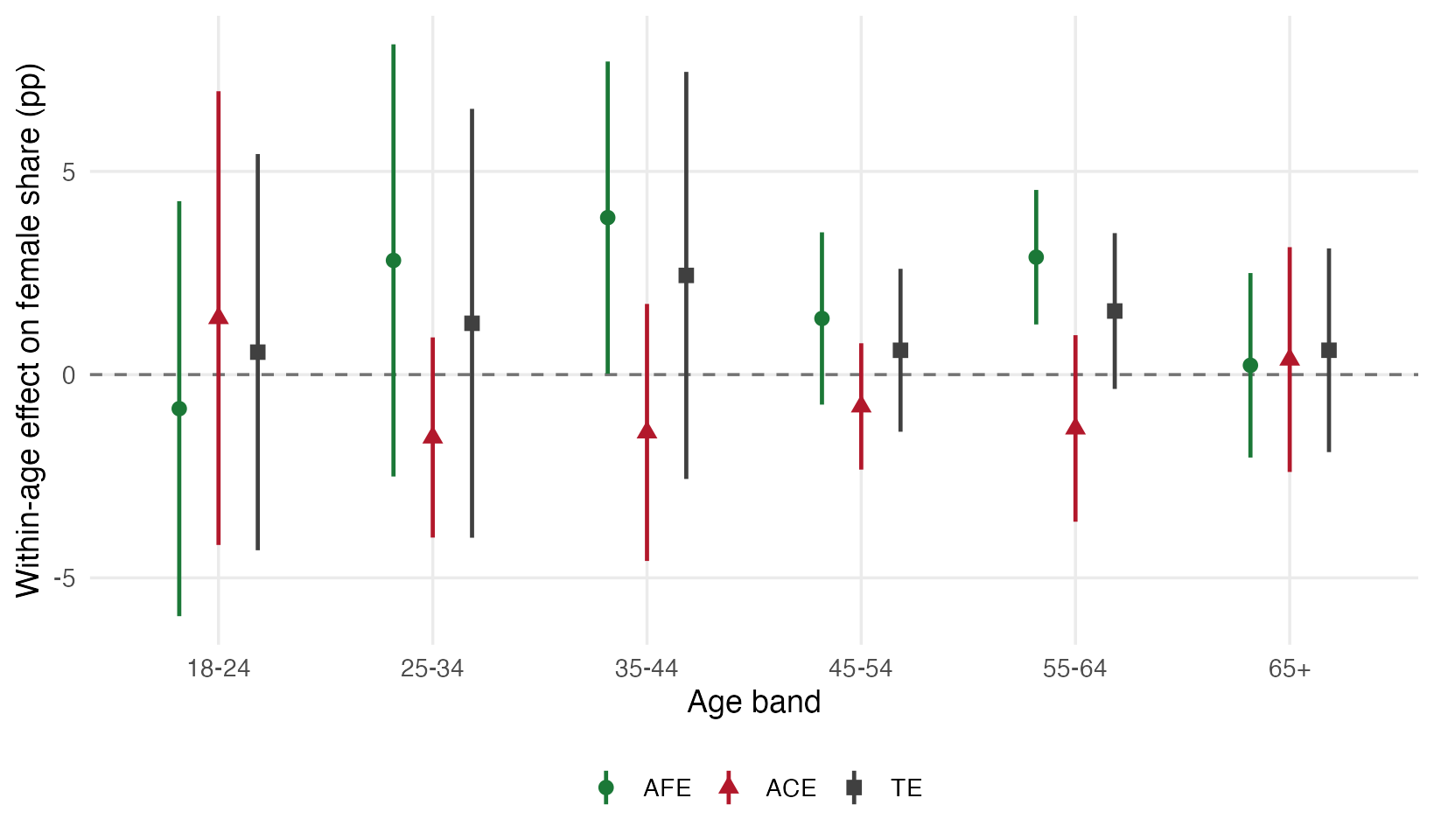}
\caption{Within-age decomposition of the female impression share by age band, high bid. Points are per-age AFE, ACE, and TE from equation~\eqref{eq:main-spec} with full (arm $\times$ age) interactions.}
\label{fig:within-age}
\end{figure}

A further disaggregation of the delivery decomposition to the twelve individual (age $\times$ gender) cells is reported in Appendix~\ref{app:further-empirical}.

\section{Implications for platforms and marketers}
\label{sec:policy}

The decomposition has consequences for three audiences. For \emph{advertisers}, a two-arm A/B test does not recover the creative's intrinsic effect; it returns the sum of the creative effect and the algorithm's targeting response. A firm that wants to know which creative is genuinely more persuasive --- independent of whom the algorithm chose to show it to --- needs the metadata-only arm, which costs one configuration field and the budget for one additional cell. The marginal cost of recovering the correct interpretation is, in other words, small.

For \emph{platforms and regulators}, the indirect effect the design isolates is, by construction, the platform's own contribution to differential exposure, separable from any choice the advertiser makes about content. This converts divergent delivery from a descriptive imbalance into a causal quantity that can be audited. Combined with proxy-based methods for inferring protected-class membership when it is unobserved \citep{kallus2022assessing}, the three-arm design supplies an audit instrument deployable even where demographic data are restricted. Three feasibility conditions apply: the auditor needs an advertiser-controlled metadata field, platform-reported demographic delivery breakdowns, and permission to run multi-arm tests. In the protected verticals that motivate the audit case --- housing, credit, employment --- platforms restrict exactly these inputs under special ad category rules, so there the design functions as a platform-internal audit tool or one compelled through regulatory data access, with proxy-based demographic inference \citep{kallus2022assessing} as the outside auditor's fallback.

Finally, the attribution carries a \emph{dynamic} risk. \citet{shimao2025strategic} argue that static fairness measures miss strategic feedback: subjects who are systematically under-served disengage, and the next generation of models trains on the resulting biased behavior. If a platform or advertiser misattributes an algorithm-driven gap to the creative, it applies the wrong remedy --- changing content while leaving the targeting mechanism untouched --- and can entrench a self-reinforcing cycle of discrimination. Our finding that the detectable reallocation is entirely algorithmic implies that the exposure disparity originates in a channel that content-side remedies do not reach, which is the channel a fairness audit must evaluate.

\section{Conclusion}
\label{sec:conclusion}

The method separates the Algorithm Filtering Effect (AFE) from the Ad Creative Effect (ACE). In the field experiment the algorithmic channel shifts female delivery share by $+2.27$ percentage points, while the creative channel is statistically indistinguishable from zero on delivery; a conventional two-arm contrast, which observes only the total, would attribute the entire shift to the wrong channel.

Theoretically, the paper extends causal-mediation design to a mediator that is an external technological artifact --- a delivery algorithm --- rather than the latent behavioral state assumed in the classical mediation literature. The substantive assumption the design rests on is verifiable from the experiment's own data within the initial window in which the serving policy is static, which is what allows the decomposition to avoid the cross-world ignorability that observational mediation requires.

Several directions follow for future work. The design assumes the algorithm's serving policy is approximately stable within the analysis window; longer panels would allow a direct test of how engagement feedback updates the serving policy over time, which we do not run here. Extending the decomposition from impression delivery to conversions, revenue, or welfare outcomes requires larger samples and longer horizons. Modeling bid pacing and auction competition structurally, rather than holding them fixed by design, would connect the decomposition to equilibrium responses. Finally, because the delivery effect sits near the design's minimum detectable effect in this campaign, a larger deployment is the natural next step to sharpen the estimates the identification strategy delivers.
\newpage
\bibliography{references}

@article{ali2019discrimination,
  author  = {Ali, Muhammad and Sapie{\.z}y{\'n}ski, Piotr and Bogen, Miranda and Korolova, Aleksandra and Mislove, Alan and Rieke, Aaron},
  title   = {Discrimination Through Optimization: How {Facebook}'s Ad Delivery Can Lead to Biased Outcomes},
  journal = {Proceedings of the ACM on Human-Computer Interaction},
  volume  = {3},
  number  = {CSCW},
  pages   = {1--30},
  year    = {2019},
  note    = {Article 199},
}

@article{aronow2017estimating,
  author  = {Aronow, Peter M. and Samii, Cyrus},
  title   = {Estimating Average Causal Effects Under General Interference, with Application to a Social Network Experiment},
  journal = {The Annals of Applied Statistics},
  volume  = {11},
  number  = {4},
  pages   = {1912--1947},
  year    = {2017},
}

@article{braun2025where,
  author  = {Braun, Michael and Schwartz, Eric M.},
  title   = {Where {A/B} Testing Goes Wrong: How Divergent Delivery Affects What Online Experiments Cannot (and Can) Tell You About How Customers Respond to Advertising},
  journal = {Journal of Marketing},
  volume  = {89},
  number  = {2},
  pages   = {71--95},
  year    = {2025},
}

@misc{burtch2025characterizing,
  author        = {Burtch, Gordon and Moakler, Robert and Gordon, Brett R. and Zhang, Poppy and Hill, Shawndra},
  title         = {Characterizing and Minimizing Divergent Delivery in {Meta} Advertising Experiments},
  year          = {2025},
  eprint        = {2508.21251},
  archivePrefix = {arXiv},
  primaryClass  = {econ.EM},
  note          = {Working paper}
}

@article{fleder2009blockbuster,
  author  = {Fleder, Daniel and Hosanagar, Kartik},
  title   = {Blockbuster Culture's Next Rise or Fall: The Impact of Recommender Systems on Sales Diversity},
  journal = {Management Science},
  volume  = {55},
  number  = {5},
  pages   = {697--712},
  year    = {2009},
}

@article{gordon2019comparison,
  author  = {Gordon, Brett R. and Zettelmeyer, Florian and Bhargava, Neha and Chapsky, Dan},
  title   = {A Comparison of Approaches to Advertising Measurement: Evidence from Big Field Experiments at {Facebook}},
  journal = {Marketing Science},
  volume  = {38},
  number  = {2},
  pages   = {193--225},
  year    = {2019},
}

@article{heckman2015econometric,
  author  = {Heckman, James J. and Pinto, Rodrigo},
  title   = {Econometric Mediation Analyses: Identifying the Sources of Treatment Effects from Experimentally Estimated Production Technologies with Unmeasured and Mismeasured Inputs},
  journal = {Econometric Reviews},
  volume  = {34},
  number  = {1--2},
  pages   = {6--31},
  year    = {2015},
}

@article{holtz2024interference,
  author  = {Holtz, David and Lobel, Felipe and Lobel, Ruben and Liskovich, Inessa and Aral, Sinan},
  title   = {Reducing Interference Bias in Online Marketplace Experiments Using Cluster Randomization: Evidence from a Pricing Meta-experiment on {Airbnb}},
  journal = {Management Science},
  volume  = {71},
  number  = {1},
  pages   = {390--406},
  year    = {2025},
}

@article{hosanagar2014village,
  author  = {Hosanagar, Kartik and Fleder, Daniel and Lee, Dokyun and Buja, Andreas},
  title   = {Will the Global Village Fracture into Tribes? {Recommender} Systems and Their Effects on Consumer Fragmentation},
  journal = {Management Science},
  volume  = {60},
  number  = {4},
  pages   = {805--823},
  year    = {2014},
}

@article{imai2013experimental,
  author  = {Imai, Kosuke and Tingley, Dustin and Yamamoto, Teppei},
  title   = {Experimental Designs for Identifying Causal Mechanisms},
  journal = {Journal of the Royal Statistical Society: Series A (Statistics in Society)},
  volume  = {176},
  number  = {1},
  pages   = {5--51},
  year    = {2013},
}

@inproceedings{imana2024apparent,
  author    = {Imana, Basileal and Korolova, Aleksandra and Heidemann, John},
  title     = {Auditing for Discrimination in Algorithms Delivering Job Ads},
  booktitle = {Proceedings of the Web Conference 2021 (WWW '21)},
  pages     = {3767--3778},
  year      = {2021},
  publisher = {ACM},
  address   = {Ljubljana, Slovenia},
}

@article{johnson2017ghost,
  author  = {Johnson, Garrett A. and Lewis, Randall A. and Nubbemeyer, Elmar I.},
  title   = {Ghost Ads: Improving the Economics of Measuring Online Ad Effectiveness},
  journal = {Journal of Marketing Research},
  volume  = {54},
  number  = {6},
  pages   = {867--884},
  year    = {2017},
}

@article{kallus2022assessing,
  author  = {Kallus, Nathan and Mao, Xiaojie and Zhou, Angela},
  title   = {Assessing Algorithmic Fairness with Unobserved Protected Class Using Data Combination},
  journal = {Management Science},
  volume  = {68},
  number  = {3},
  pages   = {1959--1981},
  year    = {2022},
}

@article{kim2025nonparametric,
  author  = {Kim, Dongwoo and Pal, Pallavi},
  title   = {Nonparametric Estimation of Sponsored Search Auctions and Impact of Ad Quality on Search Revenue},
  journal = {Management Science},
  volume  = {71},
  number  = {12},
  pages   = {10047--10066},
  year    = {2025},
}

@article{lambrecht2019algorithmic,
  author  = {Lambrecht, Anja and Tucker, Catherine},
  title   = {Algorithmic Bias? {An} Empirical Study of Apparent Gender-Based Discrimination in the Display of {STEM} Career Ads},
  journal = {Management Science},
  volume  = {65},
  number  = {7},
  pages   = {2966--2981},
  year    = {2019},
}

@article{lin2013agnostic,
  author  = {Lin, Winston},
  title   = {Agnostic Notes on Regression Adjustments to Experimental Data: Reexamining {Freedman}'s Critique},
  journal = {The Annals of Applied Statistics},
  volume  = {7},
  number  = {1},
  pages   = {295--318},
  year    = {2013},
}

@inproceedings{pearl2001direct,
  author    = {Pearl, Judea},
  title     = {Direct and Indirect Effects},
  booktitle = {Proceedings of the Seventeenth Conference on Uncertainty in Artificial Intelligence (UAI)},
  pages     = {411--420},
  year      = {2001},
  publisher = {Morgan Kaufmann},
  address   = {San Francisco, CA}
}

@article{peukert2024editor,
  author  = {Peukert, Christian and Sen, Ananya and Claussen, J{\"o}rg},
  title   = {The Editor and the Algorithm: Recommendation Technology in Online News},
  journal = {Management Science},
  volume  = {70},
  number  = {9},
  pages   = {5816--5831},
  year    = {2024},
}

@article{shimao2025strategic,
  author  = {Shimao, Hajime and Khern-Am-Nuai, Warut and Kannan, Karthik and Cohen, Maxime C.},
  title   = {Strategic Best-Response Fairness Framework for Fair Machine Learning},
  journal = {Information Systems Research},
  volume  = {36},
  number  = {4},
  pages   = {2391--2403},
  year    = {2025},
}

@book{vanderweele2015explanation,
  author    = {VanderWeele, Tyler J.},
  title     = {Explanation in Causal Inference: Methods for Mediation and Interaction},
  publisher = {Oxford University Press},
  address   = {New York},
  year      = {2015},
  isbn      = {978-0-19-932587-0}
}

@article{lewis2015unfavorable,
  author  = {Lewis, Randall A. and Rao, Justin M.},
  title   = {The Unfavorable Economics of Measuring the Returns to Advertising},
  journal = {The Quarterly Journal of Economics},
  volume  = {130},
  number  = {4},
  pages   = {1941--1973},
  year    = {2015},
}

\newpage
\appendix

\begin{center}
{\Large\bfseries Online Appendix}
\end{center}

\section{Proof of Lemma~\ref{lem:unbiased}}
\label{app:proof}

Let $Y_i(d, s)$ denote the potential outcome of unit $i$ under treatment $d \in \{A, B\}$ and mediator value $s$, and let $S_i(d)$ denote the potential mediator under treatment $d$. The observed outcome for unit $i$ in each arm is
\[
Y_i \;=\;
\begin{cases}
Y_i\!\bigl(A,\, S_i(A)\bigr) & \text{if } Z_i = 1, \\[2pt]
Y_i\!\bigl(A,\, S_i^{(2)}\bigr) & \text{if } Z_i = 2, \\[2pt]
Y_i\!\bigl(B,\, S_i(B)\bigr) & \text{if } Z_i = 3,
\end{cases}
\]
where $S_i^{(2)}$ is the realized mediator in $T_1$, which by assumption~(A3) follows the marginal distribution $F_{S(B)}$.

\paragraph{Step 1.} By construction $Y_i = Y_i\!\bigl(A, S_i(A)\bigr)$ whenever $Z_i = 1$. Therefore
\begin{align*}
\E[\bar{Y}_1]
  &= \E\!\left[\, Y_i \,\big|\, Z_i = 1 \,\right] \\
  &= \E\!\left[\, Y_i\!\bigl(A, S_i(A)\bigr) \,\big|\, Z_i = 1 \,\right] \\
  &= \E\!\left[\, Y_i\!\bigl(A, S_i(A)\bigr) \,\right] \tag{by (A1)} \\
  &= \E\!\bigl[\, Y(A, S(A)) \,\bigr].
\end{align*}

\paragraph{Step 2.} The same argument with $A$ replaced by $B$ gives
\[
\E[\bar{Y}_3] \;=\; \E\!\left[\, Y_i\!\bigl(B, S_i(B)\bigr) \,\big|\, Z_i = 3 \,\right] \;=\; \E\!\bigl[\, Y(B, S(B)) \,\bigr].
\]

\paragraph{Step 3.} Conditional on $Z_i = 2$ and $S_i^{(2)} = s$, the observed outcome is $Y_i(A, s)$ by design (the realized treatment is $A$ and the realized mediator is $s$). By the tower property of conditional expectation,
\begin{align*}
\E[\bar{Y}_2]
  &= \E\!\left[\, Y_i \,\big|\, Z_i = 2 \,\right] \\
  &= \E_{s \sim F_{S(B)}}\!\left[\, \E\!\left[\, Y_i \,\big|\, Z_i = 2,\, S_i^{(2)} = s \,\right] \,\right] \tag{tower property} \\
  &= \E_{s \sim F_{S(B)}}\!\left[\, \E\!\left[\, Y_i(A, s) \,\big|\, Z_i = 2,\, S_i^{(2)} = s \,\right] \,\right] \tag{by (A2)} \\
  &= \E_{s \sim F_{S(B)}}\!\left[\, \E\!\left[\, Y_i(A, s) \,\right] \,\right] \tag{by (A1) and (A3), independence} \\
  &= \int \E[Y(A, s)] \,\mathrm{d} F_{S(B)}(s) \tag{by (A3)} \\
  &= \E\!\bigl[\, Y(A, S(B)) \,\bigr].
\end{align*}
Each identification assumption is invoked where tagged: (A2) to replace the observed $Y_i$ given $(Z = 2, S = s)$ with the potential outcome $Y_i(A, s)$; (A1) to drop the conditioning on $Z$; and (A3) both to drop the conditioning on the realized mediator (its independence component) and to identify the marginal distribution of $S$ in $T_1$ with $F_{S(B)}$ (its distributional component).

\paragraph{Step 4.} Linearity of expectation gives
\begin{align*}
\E\!\left[\widehat{\mathrm{AFE}}\right] &= \E[\bar{Y}_2] - \E[\bar{Y}_1] = \E[Y(A, S(B))] - \E[Y(A, S(A))] = \mathrm{AFE}, \\
\E\!\left[\widehat{\mathrm{ACE}}\right] &= \E[\bar{Y}_3] - \E[\bar{Y}_2] = \E[Y(B, S(B))] - \E[Y(A, S(B))] = \mathrm{ACE}, \\
\E\!\left[\widehat{\mathrm{TE}}\right]  &= \E[\bar{Y}_3] - \E[\bar{Y}_1] = \E[Y(B, S(B))] - \E[Y(A, S(A))] = \mathrm{TE},
\end{align*}
which establishes unbiasedness of each contrast.

\paragraph{Step 5.} By inspection,
\[
\widehat{\mathrm{AFE}} + \widehat{\mathrm{ACE}} \;=\; (\bar{Y}_2 - \bar{Y}_1) + (\bar{Y}_3 - \bar{Y}_2) \;=\; \bar{Y}_3 - \bar{Y}_1 \;=\; \widehat{\mathrm{TE}},
\]
so the decomposition holds exactly in every realized sample, not merely in expectation. This completes the proof. \qed

\section{CTR decomposition and per-cell engagement analysis}
\label{app:ctr}

This appendix reports the decomposition of click-through rate (CTR), the conventional engagement outcome in ad-delivery experiments, paralleling the impression-delivery analysis in the main empirical section. CTR is defined as link clicks per impression. Because click events are rare in this campaign (mean CTR $\approx 0.08\%$), the engagement-based decomposition is substantially less well-powered than the delivery-based decomposition. We report it for completeness, and use the per-cell analysis to identify any (age $\times$ gender) cells in which the engagement contrasts are individually significant.

\subsection{Aggregate CTR decomposition}
\label{app:ctr-aggregate}

Table~\ref{tab:ctr-high-app} reports the decomposition of click-through rate at high bid under robust SE with wild-bootstrap Romano-Wolf adjusted p-values.

\begin{table}[h]
\centering
\caption{Decomposition of click-through rate, high-bid experiment. HC1 robust SEs in parentheses; $p_{\mathrm{adj}}$ are wild-bootstrap Romano-Wolf step-down adjusted across the family of three CTR contrasts (B = 5{,}000).}
\label{tab:ctr-high-app}
\begin{tabular}{lrrrr}
\toprule
Quantity & $C$ & $T_1$ & $T_2$ & $p_{\mathrm{adj}}$ \\
\midrule
CTR & $0.069\%$ & $0.096\%$ & $0.086\%$ & \\
\midrule
$\widehat{\mathrm{AFE}}$ on CTR (pp) &  & $+0.028$ \,(0.024) & & $0.532$ \\
$\widehat{\mathrm{ACE}}$ on CTR (pp) &  &  & $-0.010$ \,(0.021) & $0.734$ \\
$\widehat{\mathrm{TE}}$  on CTR (pp) &  &  & $+0.017$ \,(0.025) & $0.734$ \\
\bottomrule
\end{tabular}
\end{table}

The signs are consistent with the delivery decomposition reported in the main text: the algorithmic targeting toward women is associated with higher CTR ($\widehat{\mathrm{AFE}} = +0.028$ pp, $t = +1.28$), while the user-visible content slightly reduces CTR relative to the metadata-only condition ($\widehat{\mathrm{ACE}} = -0.010$ pp, $t = -0.48$). Neither contrast is statistically significant at the conventional threshold given the sample size ($p_{\mathrm{adj}} \geq 0.58$ for all three). The total effect is similarly small and insignificant ($+0.017$ pp, $t = +0.76$). This power gap is the principal reason we treat impression delivery, not engagement, as the primary outcome family of the analysis.

\subsection{Composition versus within-cell decomposition of the CTR effect}
\label{app:ctr-composition}

The $\widehat{\mathrm{AFE}}$ on CTR can be further decomposed into the share attributable to the algorithmic mix-shift across (age $\times$ gender) cells (the \emph{composition effect}) and the share attributable to CTR differences within those cells (the \emph{within-cell effect}):
\begin{align}
\widehat{\mathrm{AFE}}_{\mathrm{CTR}}
  &= \underbrace{\sum_g \bigl(w_{g,2} - w_{g,1}\bigr) \cdot c_{g,1}}_{\text{composition}}
   + \underbrace{\sum_g w_{g,2} \cdot \bigl(c_{g,2} - c_{g,1}\bigr)}_{\text{within-cell}},
\label{eq:comp-decomp-app}
\end{align}
where $w_{g,z}$ is the share of arm-$z$ impressions in cell $g$ and $c_{g,z}$ is the CTR of arm $z$ within cell $g$.

In the high-bid data, the composition effect accounts for approximately $1\%$ of the $\widehat{\mathrm{AFE}}$ on CTR; the within-cell effect accounts for approximately $99\%$. The algorithm's visible shift in demographic mix is therefore not the channel through which the indirect effect on CTR operates; instead, the algorithm is selecting different \emph{individuals within each demographic cell}, and those individuals respond differently to the ad. This finding indicates that demographic audits based on (age $\times$ gender) categories materially underestimate the true scale of algorithmic targeting, in line with the embedding-based imbalance findings of \citet{burtch2025characterizing}. The implication for measurement is that the impression-delivery decomposition --- which is fully observable on the demographic dimensions --- is the right primary outcome, with engagement decomposition reported as a secondary outcome whose interpretation is moderated by the existence of within-cell selection.

\subsection{Per-cell CTR decomposition by age $\times$ gender}
\label{app:ctr-cells}

Table~\ref{tab:ctr-cell-app} reports the CTR decomposition for each of the twelve (age $\times$ gender) cells at high bid. CTRs are expressed in percent; contrasts $\widehat{\mathrm{AFE}}$, $\widehat{\mathrm{ACE}}$, $\widehat{\mathrm{TE}}$ in percentage points. Standard errors are Bernoulli at the impression level, and $t$-statistics are reported in parentheses. Cells with $|t| > 1.96$ are bolded.

\begin{table}[h]
\centering
\caption{Per-cell decomposition of CTR by age $\times$ gender, high bid. Cells with $|t| > 1.96$ in bold.}
\label{tab:ctr-cell-app}
\footnotesize
\begin{tabular}{llrrrrrr}
\toprule
Age & Gender & CTR$_1$\,(\%) & CTR$_2$\,(\%) & CTR$_3$\,(\%) & $\widehat{\mathrm{AFE}}$ ($t$) & $\widehat{\mathrm{ACE}}$ ($t$) & $\widehat{\mathrm{TE}}$ ($t$) \\
\midrule
18--24 & female & 0.239 & 0.000 & 0.229 & $-0.239$ ($-1.00$) & $+0.229$ ($+1.00$) & $-0.009$ ($-0.03$) \\
18--24 & male   & 0.000 & 0.180 & 0.000 & $+0.180$ ($+1.00$) & $-0.180$ ($-1.00$) & $+0.000$ ($-$) \\
25--34 & female & 0.041 & 0.114 & 0.170 & $+0.074$ ($+0.95$) & $+0.056$ ($+0.52$) & $+0.129$ ($+1.37$) \\
25--34 & male   & 0.061 & 0.195 & 0.091 & $+0.134$ ($+1.82$) & $-0.104$ ($-1.31$) & $+0.030$ ($+0.52$) \\
35--44 & female & 0.042 & 0.104 & 0.077 & $+0.062$ ($+0.85$) & $-0.027$ ($-0.33$) & $+0.035$ ($+0.52$) \\
35--44 & male   & 0.058 & 0.057 & 0.080 & $-0.001$ ($-0.01$) & $+0.023$ ($+0.44$) & $+0.022$ ($+0.43$) \\
45--54 & female & 0.058 & 0.153 & 0.109 & $+0.096$ ($+0.90$) & $-0.044$ ($-0.38$) & $+0.051$ ($+0.53$) \\
45--54 & male   & 0.057 & 0.160 & 0.055 & $+0.103$ ($+1.34$) & $-0.105$ ($-1.38$) & $-0.002$ ($-0.03$) \\
55--64 & female & 0.211 & 0.050 & 0.143 & $-0.161$ ($-1.38$) & $+0.093$ ($+0.97$) & $-0.068$ ($-0.51$) \\
55--64 & male   & 0.079 & 0.000 & 0.050 & $-0.079$ ($-1.73$) & $+0.050$ ($+1.41$) & $-0.028$ ($-0.49$) \\
65+    & female & 0.000 & 0.129 & 0.177 & $+0.129$ ($+1.73$) & $+0.048$ ($+0.44$) & $\mathbf{+0.177}$ $\mathbf{(+2.24)}$ \\
65+    & male   & 0.106 & 0.000 & 0.000 & $\mathbf{-0.106}$ $\mathbf{(-2.00)}$ & $+0.000$ ($-$) & $\mathbf{-0.106}$ $\mathbf{(-2.00)}$ \\
\bottomrule
\end{tabular}
\end{table}

Of the 36 cell-level contrasts, three are individually significant at the $5\%$ threshold ($|t| > 1.96$): the $\widehat{\mathrm{AFE}}$ in the 65+ male cell ($-0.106$ pp, $t = -2.00$), the $\widehat{\mathrm{TE}}$ in the 65+ female cell ($+0.177$ pp, $t = +2.24$), and the $\widehat{\mathrm{TE}}$ in the 65+ male cell ($-0.106$ pp, $t = -2.00$). Both genders in the 65+ band exhibit individually significant contrasts, while no cell in the 18--54 range crosses conventional significance. Contrasts reaching $|t| > 1.65$ are $\widehat{\mathrm{AFE}}$ on 25--34 male ($+0.134$, $t = +1.82$), $\widehat{\mathrm{AFE}}$ on 55--64 male ($-0.079$, $t = -1.73$), and $\widehat{\mathrm{AFE}}$ on 65+ female ($+0.129$, $t = +1.73$); $\widehat{\mathrm{TE}}$ on 25--34 female falls marginally below the threshold ($+0.129$, $t = +1.37$).

Three substantive observations follow. First, given $3$ significant cells out of $36$ tests at the $5\%$ threshold, the cell-level CTR analysis is broadly consistent with what one would expect from chance alone: a uniform-null benchmark predicts $\sim 1.8$ significant cells. We therefore do not interpret individual cell-level CTR contrasts causally without correction for multiple testing. Second, the cells in which engagement does cross significance are concentrated in the 65+ band. This is the same band that exhibited the largest \emph{delivery} reallocation in Table~\ref{tab:cell-decomp-app}, suggesting that where the algorithm's targeting move is most pronounced, the engagement consequences are also detectable. Third, the per-cell $\widehat{\mathrm{ACE}}$ on CTR is uniformly insignificant, providing no evidence that the user-visible creative produces a within-cell engagement response that differs systematically from the metadata-only condition --- consistent with the aggregate finding that engagement effects of the visible content are difficult to detect at this sample size.

A Bonferroni correction at the family level ($36$ tests, $\alpha = 0.05$) yields a per-test threshold of $|t| > 3.06$. None of the cell-level CTR contrasts crosses this threshold. We therefore conclude that the cell-level engagement decomposition does not by itself identify any (age $\times$ gender) cell as causally affected by the treatment after multiple-testing correction, and that engagement-based audits of algorithmic differential treatment at the scale of this experiment would require either a substantially larger sample or pooling across cells along the lines of Section~\ref{sec:empirical-delivery}.

\subsection{Simulation analogue: realized vs predicted CTR}
\label{app:ctr-sim}

This subsection ports the engagement-side analysis of the simulation in Section~\ref{sec:simulation} to provide a structural reading of the small empirical CTR effects reported in Appendix~\ref{app:ctr-aggregate}. The simulation records two engagement quantities per focal impression: the \emph{realized} click rate $\hat{\mu}_r = (1/W) \sum_{u: W(u) = 0} r_{u, 0}$ (the true click probability averaged over focal's wins) and the \emph{predicted} click rate $\hat{\mu}_s = (1/W) \sum_{u: W(u) = 0} s_{u, 0}$ (the platform's signal averaged over focal's wins). At the calibrated $\sigma = 0.030$ with the focal bid at the 99.5th percentile of $F_v$, the per-arm averages across $R = 1{,}000$ Monte-Carlo rounds are reported in Table~\ref{tab:sim-ctr}.

\begin{table}[h]
\centering
\caption{Simulated realized and predicted CTR over focal impressions, by arm.}
\label{tab:sim-ctr}
\begin{tabular}{lrrr}
\toprule
Quantity & $C$ & $T_1$ & $T_2$ \\
\midrule
Realised CTR $\hat{\mu}_r$         & $0.0721$ & $0.0515$ & $0.0515$ \\
Predicted CTR $\hat{\mu}_s$        & $0.0721$ & $0.0736$ & $0.0736$ \\
Predicted $-$ realized gap         & $0.0000$ & $+0.0221$ & $+0.0221$ \\
\bottomrule
\end{tabular}
\par\smallskip
\footnotesize\emph{Note.} Click rates in raw proportions; differences are discussed in the text in percentage points.

\end{table}

Under the clean signal regime ($C$) predicted and realized CTRs coincide exactly, since $s = r$. Under the noisy regime ($T_1$ and $T_2$) the platform's predicted CTR rises slightly while the realized CTR drops by approximately $29\%$ relative to $C$, a classical winner's-curse pattern in which selecting impressions on a noisy signal pulls in low-$r$ users whose signals were inflated by chance. The gap between predicted and realized CTR, about $0.022$, is the magnitude of the prediction error a naive analyst would make by reading $\hat{\mu}_s$ off the platform's reporting interface and trusting it as realized performance.

Table~\ref{tab:sim-ctr-decomp} reports the corresponding AFE, ACE, and TE on both realized and predicted CTR.

\begin{table}[h]
\centering
\caption{Simulated AFE, ACE, and TE on realized and predicted CTR, mean (SE) with 95\% CI.}
\label{tab:sim-ctr-decomp}
\begin{tabular}{lccc}
\toprule
Names & Estimate & SE & 95\% CI \\
\midrule
$\widehat{\mathrm{AFE}}$ on $\mu_r$ & $-0.0206$ & $0.0000$ & $[-0.0207,\, -0.0206]$ \\
$\widehat{\mathrm{ACE}}$ on $\mu_r$ & $-0.0000$ & $0.0000$ & $[-0.0001,\, +0.0000]$ \\
$\widehat{\mathrm{TE}}$  on $\mu_r$ & $-0.0206$ & $0.0000$ & $[-0.0207,\, -0.0206]$ \\
\midrule
$\widehat{\mathrm{AFE}}$ on $\mu_s$ & $+0.0015$ & $0.0000$ & $[+0.0014,\, +0.0015]$ \\
$\widehat{\mathrm{ACE}}$ on $\mu_s$ & $+0.0000$ & $0.0000$ & $[-0.0000,\, +0.0001]$ \\
$\widehat{\mathrm{TE}}$  on $\mu_s$ & $+0.0015$ & $0.0000$ & $[+0.0015,\, +0.0015]$ \\
\bottomrule
\end{tabular}
\par\smallskip
\footnotesize\emph{Note.} Click rates in raw proportions; differences are discussed in the text in percentage points.

\end{table}

In the simulation the indirect channel cuts realized CTR by about $-2.1$ percentage points and raises predicted CTR by about $+0.2$ percentage points. The empirical CTR decomposition of Appendix~\ref{app:ctr-aggregate} finds a small, statistically insignificant AFE on CTR ($+0.028$ pp, $p_{\mathrm{adj}} = 0.576$). The simulation supplies a candidate explanation. At the signal-noise level implied by the empirical delivery AFE,\footnote{A back-of-envelope calibration: scaling $\sigma = 0.030$ by the ratio of empirical to simulated delivery AFE, $2.27/16.2$, gives $\sigma \approx 0.004$.} the matching channel's engagement-side bite is much smaller than its delivery-side bite, and finite-sample variability easily overwhelms it. The simulation under no-creative-effect predicts a \emph{negative} sign on the realized-CTR AFE (the winner's-curse direction), while the empirical estimate is small and positive; the discrepancy is within both the empirical CI and the range of additional primitives --- a positive creative effect on $r$, or non-additive signal noise --- that the simulation could accommodate. We report the simulated engagement decomposition as a reading aid for the empirical pattern of Appendix~\ref{app:ctr-aggregate}, not as a competing estimate.

\clearpage

%
%
%

\section{Further details on the simulation}
\label{app:sim-details}

This appendix collects supporting material for the simulation of
Section~\ref{sec:simulation}. Section~\ref{app:sim-primitives} lays out the
population and click-rate primitives in full. Section~\ref{app:sim-bids}
explains the choice of bid levels for the focal advertiser.
Section~\ref{app:sim-split} explains the split of non-focal advertisers
across the two ad types and its role in generating heteroskedastic signal
noise.

\subsection{Population, click-rate, and value primitives}
\label{app:sim-primitives}

The platform has $N_i$ advertisers and $N_u$ users. Each advertiser $i$ has
a per-click value $v_i$ drawn i.i.d.\ from a distribution $F_v$, and each
user $u \in \{1, 2, \dots, N_u\}$ has a gender $g_u \in \{M, F\}$. User
preferences for advertisers are captured through \emph{click rates} --- the
probability that user $u$ would click on advertiser $i$'s ad. True click
rates are $r_{u,i} \in [0,1]$, drawn i.i.d.\ across $(u, i)$ from a common
Beta distribution:
\[
r_{u,i} \;\sim\; \mathrm{Beta}(\alpha, \beta), \qquad \E[r] = \mu_r.
\]

\paragraph{Why a Beta distribution.}
The Beta family delivers an L-shape: most users have low click probability
with a fat right tail. This matches the empirical shape of digital-ad CTR
distributions documented in \citet{kim2025nonparametric}, and is the reason
we use it in place of, say, a truncated Gaussian.

\paragraph{Advertiser expected return.}
An advertiser $i$'s expected return from serving user $u$ is
\[
\mathrm{ER}_{u,i} \;=\; v_i \cdot r_{u,i}.
\]
Bids are set at $b_i = v_i$ (truth-telling under second-price sealed-bid
auctions) so the auction ranks advertisers by valuation, not expected
return; the expected return enters the outcome only through the realized
click.

\paragraph{Ad-type notation.}
Each advertiser $i$ has an ad type $a(i)$. The ad type interacts with the
user's gender to determine the variance of the platform's matching-signal
noise, $\sigma^2_{a(i), g_u}$; this is what makes the algorithm-side signal
heteroskedastic across (ad-type $\times$ gender) cells and is the sole
channel through which the algorithm produces AFE.

\subsection{Choice of bid levels for the focal advertiser}
\label{app:sim-bids}

The focal advertiser's bid is swept across three top-tail percentiles of
$F_v$: the 95th, 98th, and 99.5th. Quartile bids (25th, 50th, 75th
percentile) yield essentially zero wins in the simulation and therefore
zero identification of the impression-stage AFE --- the focal ad needs to be
in the upper tail of $F_v$ to ever win the second-price auction against
the $N_i - 1$ incumbents.

The three chosen bid levels span roughly $80$ to $780$ expected wins per
simulation round, which is enough for tight Monte-Carlo error on the
impression-stage effect at every bid tier and enough to trace out the
relationship between win rate and how much of the match-stage bias
propagates into what the advertiser actually sees.

\subsection{Split of non-focal advertisers across ad types}
\label{app:sim-split}

We assume the $N_i - 1$ non-focal (``incumbent'') advertisers are split
near-equally across the two ad types. In the reported specification the
split is 250 Ad A advertisers and 249 Ad B advertisers (with the focal
being the 500th).

This near-equal split serves two purposes. First, it keeps the two ad
types symmetric in the population, so the ad-type $\times$ gender noise
structure is the only source of asymmetry between $A$ and $B$. Second, it
means that both $r_{u,i}$ and $s_{u,i}$ are advertiser-specific: the ad
type interacts with the advertiser identity to produce the pair
$(r_{u,i}, s_{u,i})$, which in turn drives the auction outcome. If we
instead set all incumbents to a single ad type, the focal advertiser's
signal profile would be identifiable from the incumbent baseline, and the
three-arm construction would collapse.


\clearpage

%
%
%

\section{The effect of the focal bid: how the matching-stage algorithm bias transfers to the bid-winner stage}
\label{app:bid-transfer}

The main text establishes the algorithm bias as a property of the matching
stage: at $\sigma = 0.030$, the focal's matched pool is $53.7\%$ female in
arm $C$ and $69.9\%$ female in arm $T_1$, giving
$\widehat{\mathrm{AFE}}^{\mathrm{match}} = +16.22$ ppt with a tight
Monte-Carlo confidence interval. The advertiser, however, does not observe
matches. She observes \emph{impressions}: the subset of matches that
survive the second-price auction. This appendix formalizes the distinction
between the two measures, states the match-stage invariance lemma, and
documents how the matching-stage bias transfers into the impression-stage
estimand across a range of focal bid levels. It closes with a
design-integrity check confirming that the arms do not distort the
platform's overall serve rate.

\subsection{Two natural measures of the algorithm's bias}
\label{app:two-measures}

The algorithm's targeting bias lives at the matching stage: $M(u)$ is
where the platform decides which advertisers compete for user $u$, and
this decision uses only the signal, not any bid. The focal advertiser,
however, only sees the impressions it wins, which are a filtered sample of
$M(u)$ picked by the second-price auction. So there are two natural
fractions to look at in each arm $z \in \{C, T_1, T_2\}$:
\begin{itemize}
\item $\pi^{\text{match}}_F(z)$: the fraction of the focal's matched users
that are female;
\item $\pi^{\text{imp}}_F(z, b_0)$: the fraction of the focal's impressions
that are female, for a given focal bid $b_0$.
\end{itemize}
The first is a property of the top-$qN_i$ selection alone; the second is
what the advertiser can actually see. This gives us two versions of AFE,
one at the match stage and one at the impression stage.

\subsection{Match-stage bias does not depend on bids}
\label{app:lemma-match-invariance}

The following result formalizes the point that the match-stage AFE is a
structural property of the signal, independent of both the focal's bid and
the incumbent bid distribution.

\begin{lemma}[Match-stage bias does not depend on bids]
\label{lem:sim-match-invariance}
The match-stage AFE, $\pi^{\text{match}}_F(T_1) - \pi^{\text{match}}_F(C)$,
takes the same value for every choice of the focal bid $b_0$ and every
choice of the bid distribution $F_v$.
\end{lemma}
\begin{proof}
The matched pool $M(u)$ is decided only by the signal values
$s_{u,i} = r_{u,i} + \varepsilon_{u,i}$, and neither $r_{u,i}$ nor
$\varepsilon_{u,i}$ uses any bid. So the joint distribution of
$\bigl(g_u,\ \mathbf{1}[0 \in M(u)]\bigr)$ under any given arm is the same
no matter what the bids are. Hence the two probabilities that define
$\pi^{\text{match}}_F(z)$,
\[
\pi^{\text{match}}_F(z) \;=\; \Prb\bigl(g_u = F \,\bigm|\, 0 \in M(u),\ \text{arm} = z\bigr),
\]
take the same value for every choice of the focal bid $b_0$ and every bid
distribution $F_v$. Their difference
$\pi^{\text{match}}_F(T_1) - \pi^{\text{match}}_F(C)$ therefore does too.
\end{proof}

\subsection{Impression-stage AFE across bid percentiles}
\label{app:bid-transfer-chart1b}

Figure~\ref{fig:sim-1b} plots the impression-stage
$\widehat{\mathrm{AFE}}$, $\widehat{\mathrm{ACE}}$, and
$\widehat{\mathrm{TE}}$ across the three focal bid levels of the full run
--- the 95th, 98th, and 99.5th percentiles of the incumbent bid
distribution $F_v$. The impression-stage $\widehat{\mathrm{AFE}}$ is
$+16.40$ (SE $0.12$), $+16.21$ (SE $0.05$), and $+16.21$ (SE $0.03$)
percentage points respectively: statistically indistinguishable from the
match-stage anchor of $+16.22$ ppt at every level.

At the 95th percentile
the focal wins only about $8\%$ of the users it is matched to, so the
estimate rests on far fewer impressions and its Monte-Carlo standard error
is four times larger --- but the point estimate is already at the anchor.
Because the auction's win filter is gender-neutral conditional on
matching, the impression-stage estimator is unbiased for the match-stage
AFE at any bid; a low bid produces a less precise estimate.

\begin{figure}[h]
\centering
\includegraphics[width=0.75\textwidth]{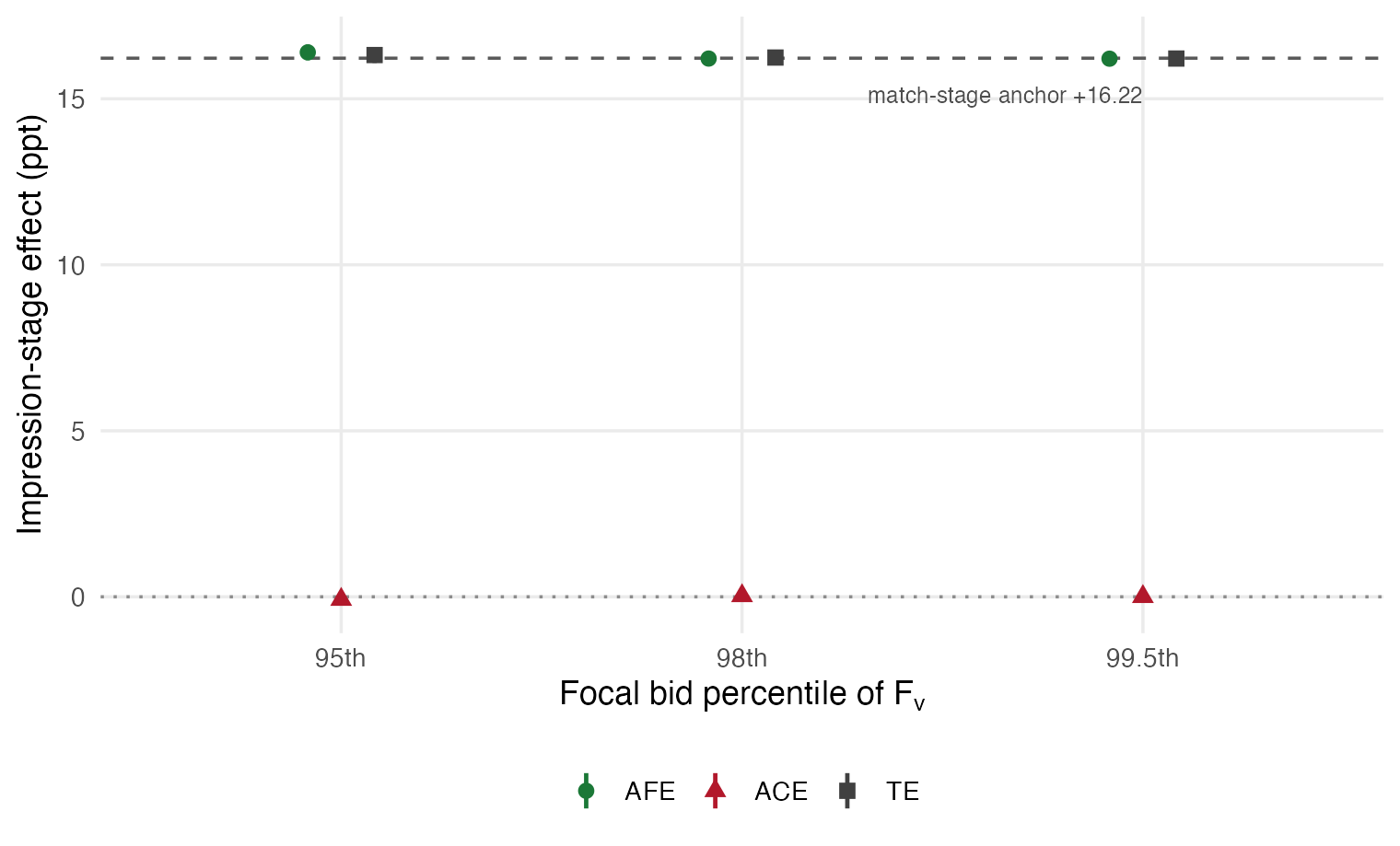}
\caption{Impression-stage $\widehat{\mathrm{AFE}}$, $\widehat{\mathrm{ACE}}$, and $\widehat{\mathrm{TE}}$ across focal bid levels (95th, 98th, 99.5th percentiles of $F_v$).}
\label{fig:sim-1b}
\end{figure}

\subsection{The visibility ratio: how much of the bias the advertiser sees}
\label{app:bid-transfer-chartB}

Figure~\ref{fig:sim-B} recasts the same content as a ratio: the
impression-stage AFE divided by the match-stage AFE. This ratio is a
measure of how much of the algorithm's structural bias becomes visible in
the impressions the focal wins.

\begin{figure}[h]
\centering
\includegraphics[width=0.75\textwidth]{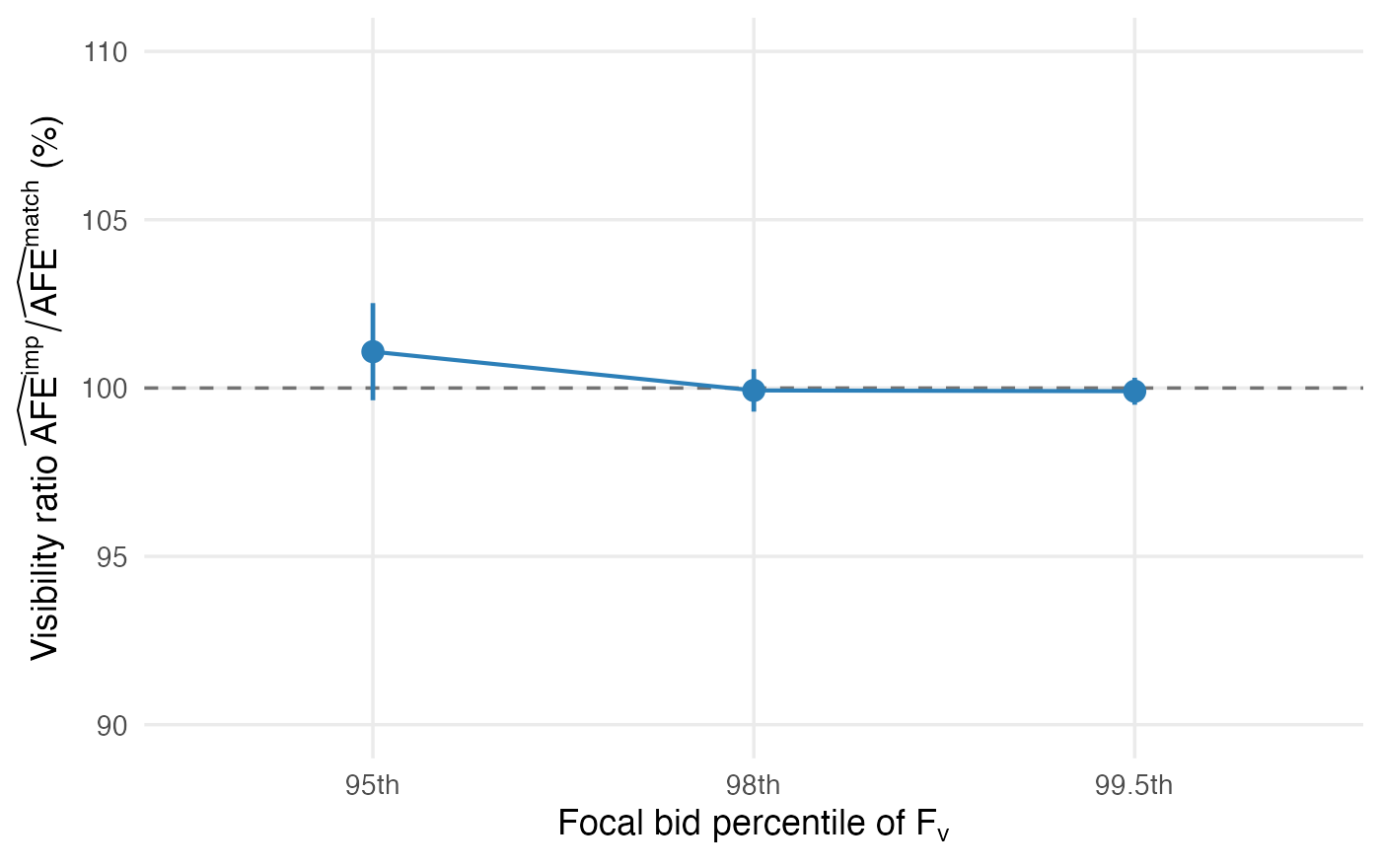}
\caption{Visibility ratio $\widehat{\mathrm{AFE}}^{\mathrm{imp}} / \widehat{\mathrm{AFE}}^{\mathrm{match}}$ across the same bid levels.}
\label{fig:sim-B}
\end{figure}

\subsection{Design integrity: no arm-level interference}
\label{app:no-interference}

Figure~\ref{fig:sim-C} plots the focal's wins per round in all
three arms across the same three bid levels. The $T_1$ and $T_2$ lines sit
on top of each other at every bid, and $C$ sits below them by the amount
the matching bias predicts. This indicates that the arms only reshuffle
impressions among matched users; they do not change how many ads the
platform serves overall. In other words, the three-arm design cleanly
identifies AFE and ACE without contaminating the total serve rate.

\begin{figure}[h]
\centering
\includegraphics[width=0.75\textwidth]{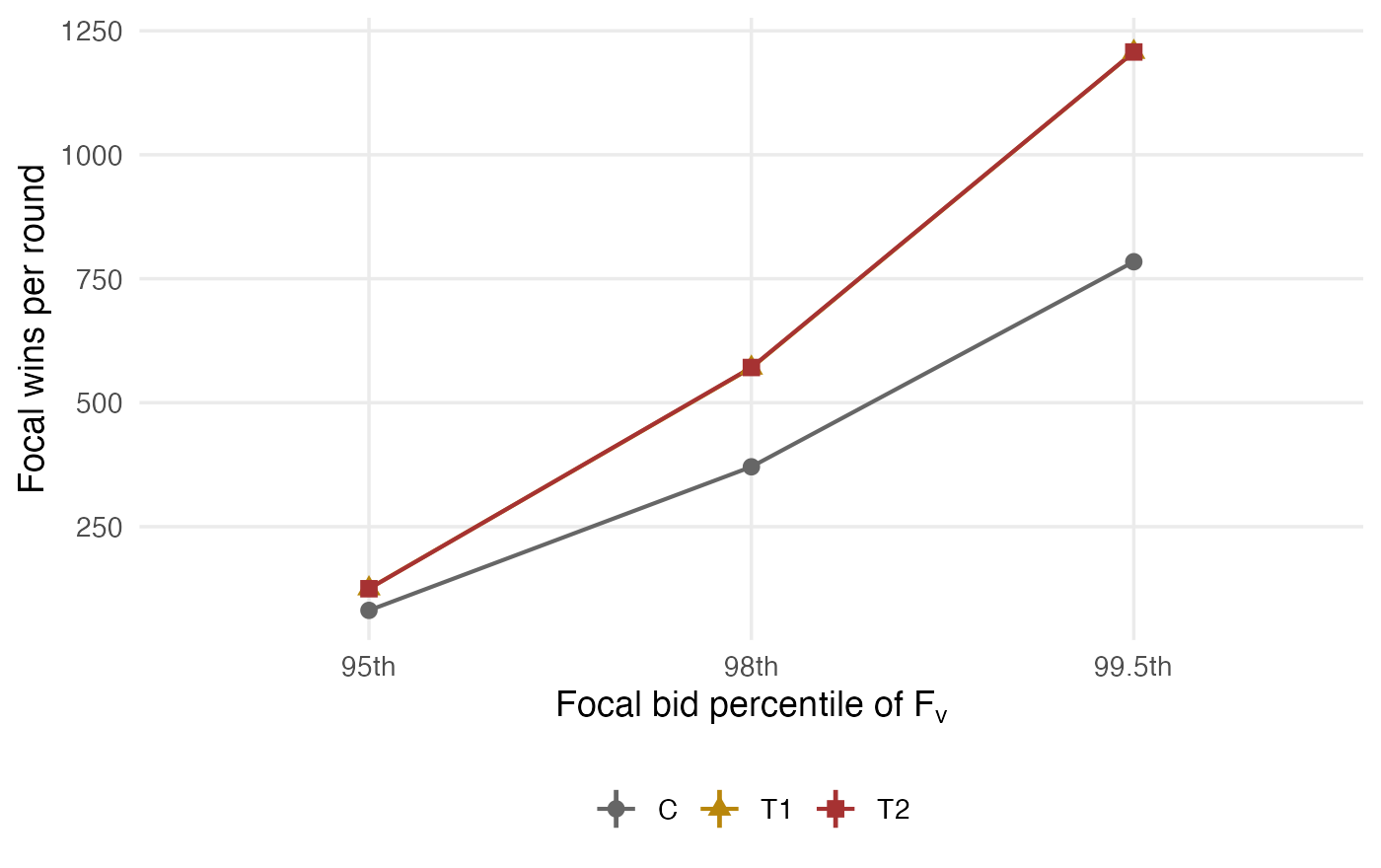}
\caption{Focal's wins per round in each arm, across three bid percentiles.}
\label{fig:sim-C}
\end{figure}

\subsection{Interpretation: structural bias vs.\ visibility}
\label{app:bid-transfer-interp}

The three charts together deliver a single interpretation. The
algorithm's targeting bias is fixed at $+16.2$ ppt
(Lemma~\ref{lem:sim-match-invariance}): it is a structural property of the
matching stage and it does not depend on the bid. What varies with
the bid is only the \emph{precision} with which that bias is estimated
from impression data: a low-bidding focal wins few auctions, so the
impression-stage estimate rests on a small sample and its Monte-Carlo
standard error widens ($0.12$ at the 95th percentile versus $0.03$ at the
99.5th). The point estimate itself does not move: even at the 99.5th
percentile the focal wins roughly $78\%$ of its matched users, yet the
impression-stage estimate already coincides with the match-stage anchor
at every level, because the second-price filter is gender-neutral
conditional on matching.

\paragraph{Empirical implication.}
The flat bid profile has a direct empirical consequence: two
advertisers running the identical creative at different percentiles of
$F_v$ measure the \emph{same} impression-stage AFE in expectation --- the
lower bidder simply measures it less precisely. The three-arm estimator
applied at the impression level therefore recovers the match-stage
structural bias at any bid position; cross-experiment comparisons need
alignment on bid position only for the precision of the estimate.

At the 95th, 98th, and 99.5th percentiles of $F_v$, the simulated
impression-stage $\widehat{\mathrm{AFE}}$ is $+16.40$ ($0.12$), $+16.21$
($0.05$), and $+16.21$ ($0.03$) percentage points respectively, with the
$\widehat{\mathrm{ACE}}$ statistically zero at every level. The
algorithm bias is thus operative, and measurable from impression data,
independently of the bid --- matching uses the signal and not the bid, and
the auction passes the matched composition through to wins essentially
undistorted; what rises with the bid is the precision of the estimate.


\clearpage

\section{The ACE in the simulation across bid levels}
\label{app:sim-nde-placebo}

Under the no-creative-effect setting of Section~\ref{sec:simulation}, both versions of the ACE come back at zero, and the estimates agree. At the match stage, $\widehat{\mathrm{ACE}}^{\text{match}} = +0.00$ ppt with 95\% CI $[-0.04,\ +0.04]$. At the impression stage, the ACE stays inside a narrow band of zero at every bid level; the tightest reading is at the 99.5th percentile of $F_v$, where the CI is $[-0.04,\ +0.05]$. This gives us two useful checks at once: the simulator is not fabricating a creative effect, and the paired-difference estimator is not picking up a spurious signal.

\begin{figure}[h]
\centering
\includegraphics[width=0.75\textwidth]{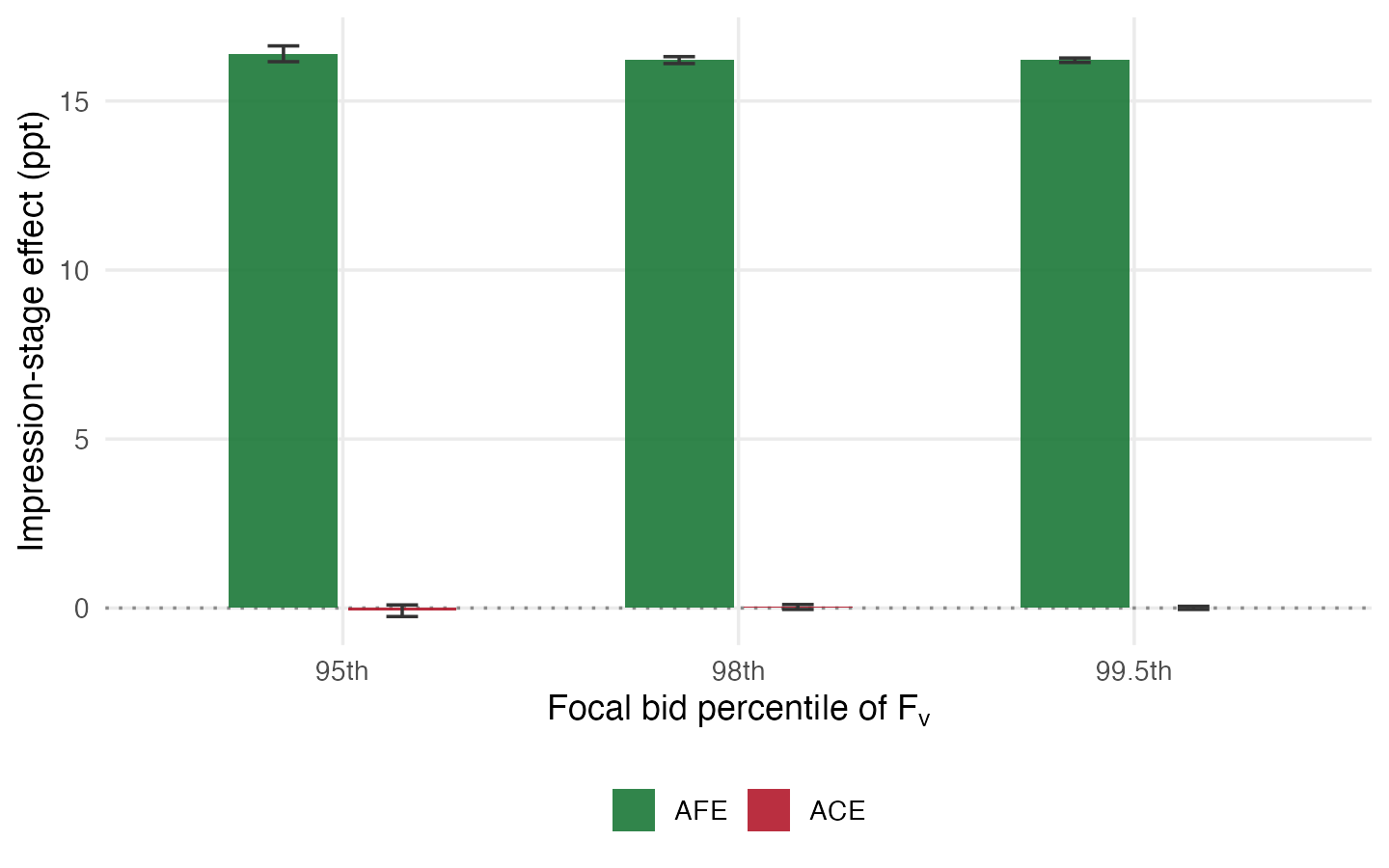}
\caption{Point estimates of $\widehat{\mathrm{AFE}}$ (algorithm filtering channel) and $\widehat{\mathrm{ACE}}$ (creative channel) at three focal bid percentiles (95th, 98th, 99.5th) with 95\% CIs.}
\label{fig:sim-1a}
\end{figure}

\section{Additional simulation figures: gender and click-rate heterogeneity}
\label{app:sim-heterogeneity}

Figures~\ref{fig:sim-2a}--\ref{fig:sim-2c} decompose the algorithm channel by user gender. Figures~\ref{fig:sim-3a}--\ref{fig:sim-3c} decompose it further by (click-rate quintile $\times$ gender), where the quintiles are ordered by the average predicted click rate $\bar r_u$ across all advertisers. Together they show that the matching-channel effect is concentrated entirely on female users and is roughly uniform across click-rate quintiles: the bias is a treatment on gender, not a differential response by user quality.

\begin{figure}[h]
\centering
\begin{subfigure}[t]{0.48\textwidth}
\centering
\includegraphics[width=\textwidth]{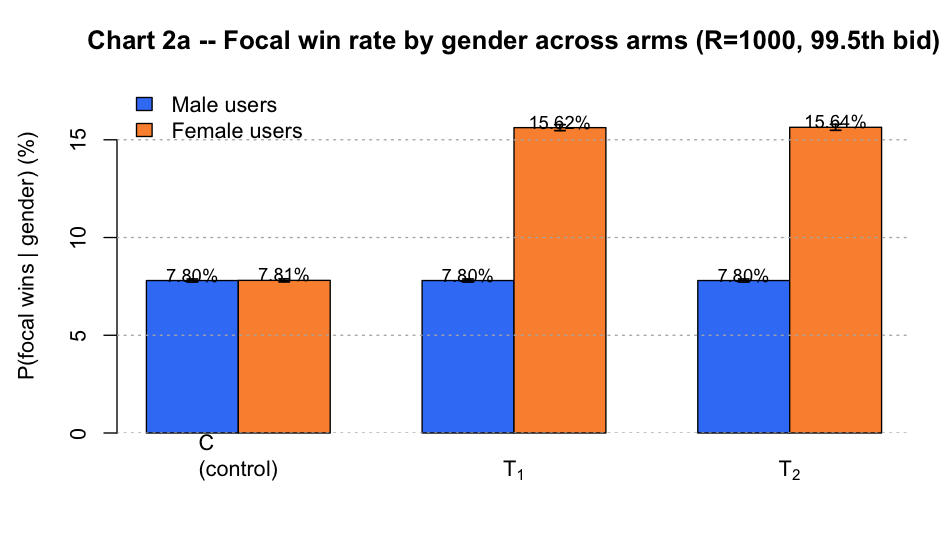}
\caption{$P(\text{focal wins}\mid\text{gender})$ by arm, at the 99.5th percentile bid.}
\label{fig:sim-2a}
\end{subfigure}\hfill
\begin{subfigure}[t]{0.48\textwidth}
\centering
\includegraphics[width=\textwidth]{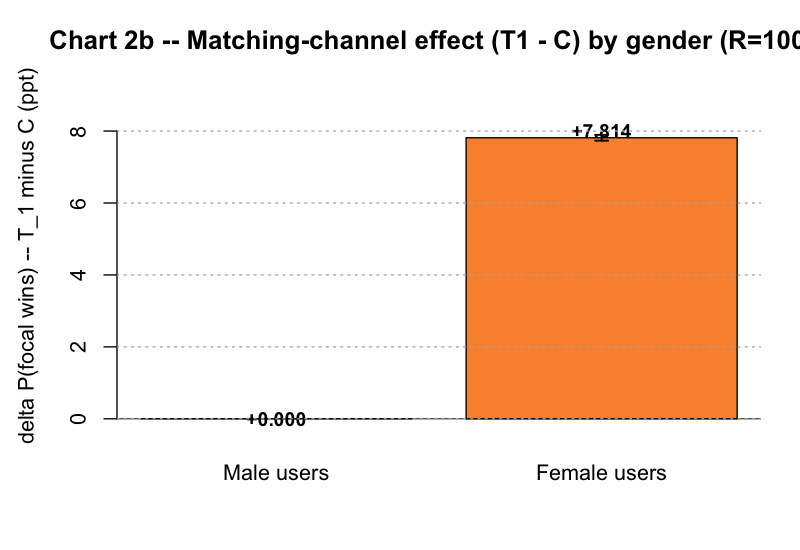}
\caption{Matching-channel effect ($T_1 - C$) on the win rate, by gender.}
\label{fig:sim-2b}
\end{subfigure}

\vspace{0.6em}

\begin{subfigure}[t]{0.62\textwidth}
\centering
\includegraphics[width=\textwidth]{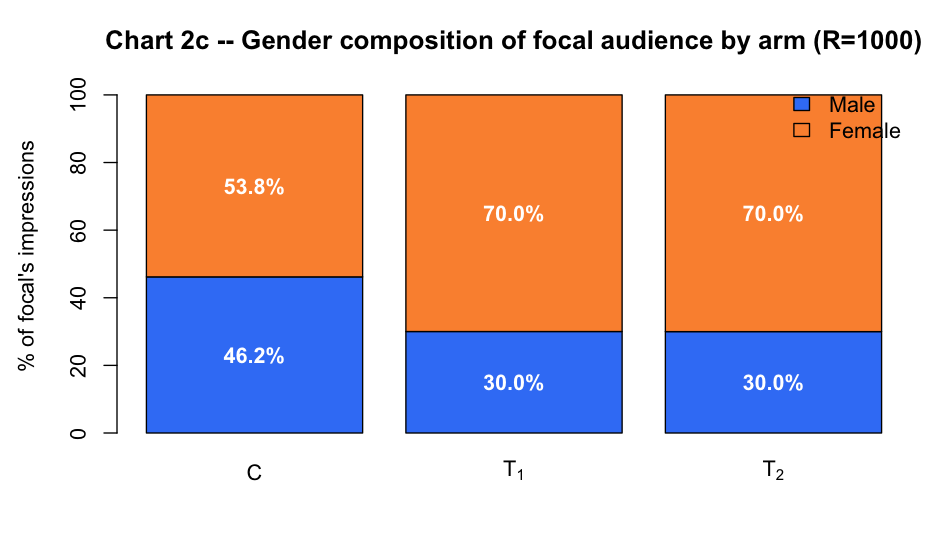}
\caption{Gender composition of the focal's audience by arm.}
\label{fig:sim-2c}
\end{subfigure}
\caption{Gender heterogeneity of the matching channel.}
\label{fig:sim-gender-het}
\end{figure}

\begin{figure}[h]
\centering
\begin{subfigure}[t]{0.48\textwidth}
\centering
\includegraphics[width=\textwidth]{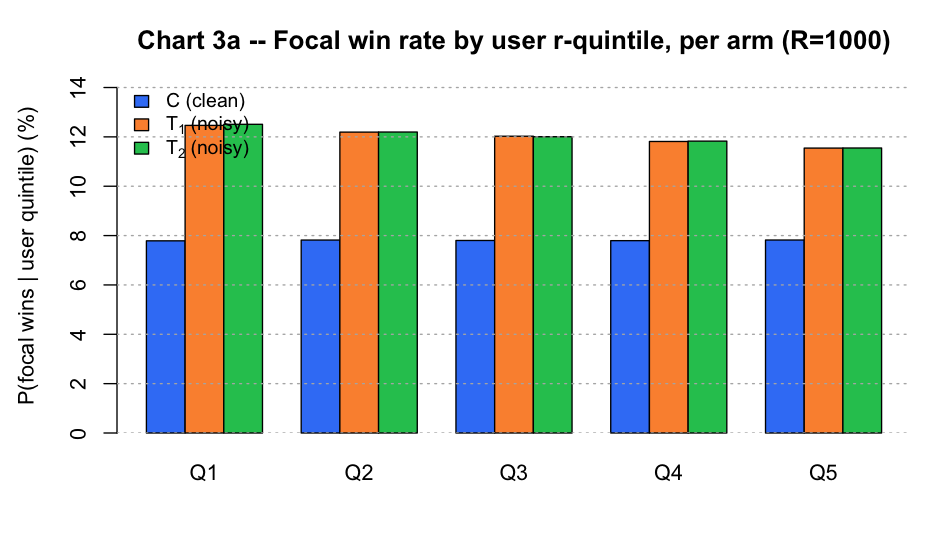}
\caption{Focal win rate by user click-rate quintile ($\bar r_u$), per arm.}
\label{fig:sim-3a}
\end{subfigure}\hfill
\begin{subfigure}[t]{0.48\textwidth}
\centering
\includegraphics[width=\textwidth]{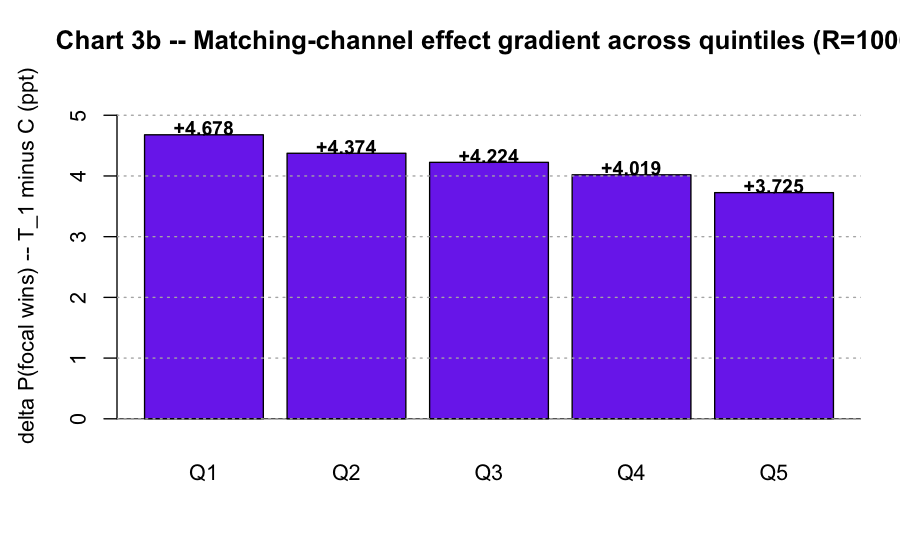}
\caption{Matching-channel effect ($T_1 - C$) by click-rate quintile, pooled across gender.}
\label{fig:sim-3b}
\end{subfigure}

\vspace{0.6em}

\begin{subfigure}[t]{0.62\textwidth}
\centering
\includegraphics[width=\textwidth]{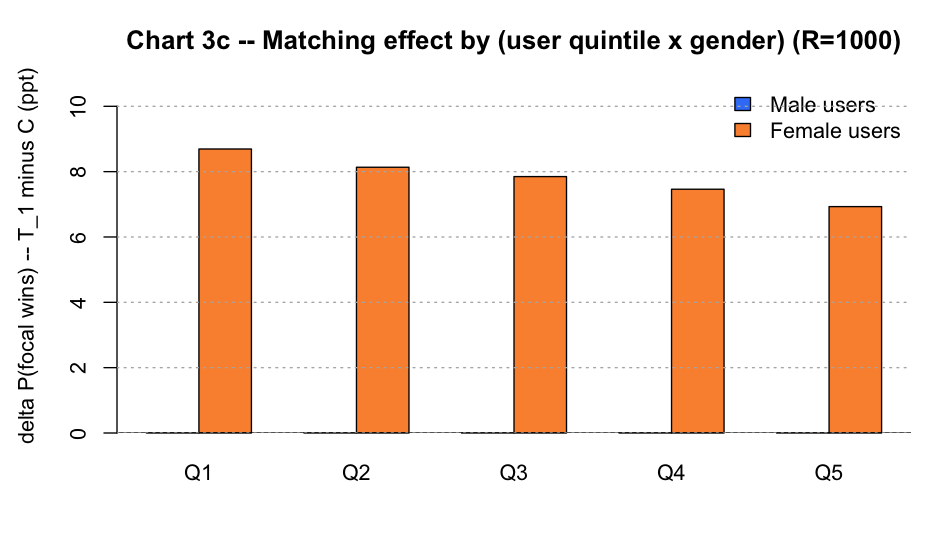}
\caption{Matching-channel effect by (click-rate quintile $\times$ gender).}
\label{fig:sim-3c}
\end{subfigure}
\caption{Click-rate quintile heterogeneity of the matching channel.}
\label{fig:sim-r-het}
\end{figure}

\clearpage

\section{Wild-bootstrap Romano-Wolf adjustment}
\label{app:wildboot}

The main-text tables report analytical clustered (CR2) standard errors on point estimates from a design with 19 (day $\times$ time-slot) clusters at high bid. The finite-sample distribution of the clustered $t$-statistic may deviate from its asymptotic approximation, and the three pairwise contrasts AFE, ACE, and TE constitute a family for which joint significance requires family-wise error control.

We construct a robustness check by resampling the empirical residual distribution. For $B = 5{,}000$ iterations, we draw Rademacher weights $\omega_b \in \{-1, +1\}$ independently at the (day $\times$ time-slot) cluster level, apply each cluster's weight to all of its rows, form $y^*_i = \hat{y}_i + \omega_{b(i)} \hat{\varepsilon}_i$, refit the specification of Section~\ref{sec:inference}, recompute clustered standard errors, and record the centered $t$-statistic $t^*_{b,k} = (\hat{\beta}^*_{b,k} - \hat{\beta}_k)/\widehat{\mathrm{SE}}(\hat{\beta}^*_{b,k})$. Cluster-level (wild cluster) resampling is the correct pairing for the clustered analytic standard errors. The Romano-Wolf step-down procedure controls the family-wise error rate across the three pairwise contrasts, accounting for their joint dependence structure.

\begin{table}[h]
\centering
\caption{Wild-bootstrap Romano-Wolf adjusted $p$-values on the aggregate decomposition, high-bid experiment.}
\label{tab:wildboot}
\begin{tabular}{lrrrr}
\toprule
Contrast & Estimate (SE) & Analytical $p$ & Bootstrap $p_{\mathrm{unadj}}$ & Romano-Wolf $p_{\mathrm{adj}}$ \\
\midrule
$\widehat{\mathrm{AFE}}$ & $+2.27$ \,(0.77) & $0.003$ & $0.018$ & $0.069$ \\
$\widehat{\mathrm{ACE}}$ & $-0.93$ \,(0.56) & $0.095$ & $0.152$ & $0.240$ \\
$\widehat{\mathrm{TE}}$  & $+1.33$ \,(1.01) & $0.186$ & $0.236$ & $0.240$ \\
\bottomrule
\end{tabular}
\par\smallskip
\footnotesize\emph{Note.} Standard errors in parentheses. $B = 5{,}000$.
\end{table}

The bootstrap procedure is more conservative than the analytical clustered $t$-test with 19 clusters: the Romano-Wolf adjusted $p$-value for the algorithm-channel AFE is $0.069$, compared with $0.003$ under the analytical test. Point estimates are unchanged.

\clearpage

\section{Further details on the empirical analysis}
\label{app:further-empirical}

\subsection{Campaign configurations held constant across arms}
\label{app:campaign-config}

The following configurations are identical across the three arms within each bid level. This structural symmetry ensures that any cross-arm difference in impression composition is attributable to the ad creative (the intended manipulation) rather than to campaign-configuration heterogeneity of the sort documented by \citet{burtch2025characterizing}.

\begin{itemize}
\item Campaign objective (traffic) and optimization goal.
\item Target audience definition (geography, age range, broad interests).
\item Campaign budget cap per cell.
\item Ad placement: manual, restricted to a single surface (Facebook news feed) to preclude placement-driven divergent delivery \citep{burtch2025characterizing}.
\item Frequency cap of approximately one impression per user (observed impression-to-reach ratio $\approx 1.009$).
\item Schedule and duration.
\item Ad format: static image.
\end{itemize}

The ad creative and the bid amount are the two variables intentionally varied. The ad creative is the manipulation of interest; the bid amount is a secondary source of heterogeneity retained for robustness (Appendix~\ref{app:bid-heterogeneity}).

\subsection{Descriptive gender breakdown and supporting figures}
\label{app:descriptive-extras}

\subsubsection*{Gender breakdown of arm-level totals}

Table~\ref{tab:gender} disaggregates the arm-level totals of Table~\ref{tab:armlevel} by gender. The female-minus-male difference in impressions is stable across arms ($\Delta \approx -9{,}000$), consistent with the algorithm delivering more impressions to men across all three arms. The corresponding difference in CTR widens from $0.000$ pp in $C$ to $+0.084$ pp in $T_2$, and the female CPM premium widens from $+\$0.04$ in $C$ to $+\$0.16$ in $T_2$.

\begin{table}[h]
\centering
\small
\caption{Gender breakdown, high-bid arms (gender $\in \{$female, male$\}$).}
\label{tab:gender}
\begin{tabular}{llrrrrr}
\toprule
Arm & Gender & Impressions & Reach & Clicks & CTR (\%) & CPM (\$) \\
\midrule
$C$   & Female      & 11{,}595 & 11{,}515 &   8 &  0.069 & 4.099 \\
$C$   & Male        & 21{,}730 & 21{,}527 &  15 &  0.069 & 4.062 \\
$C$   & $\Delta$(F$-$M) & $-10{,}135$ & $-10{,}012$ & $-7$ & $+0.000$ & $+0.037$ \\
\midrule
$T_1$ & Female      & 12{,}267 & 12{,}149 &  13 &  0.106 & 4.184 \\
$T_1$ & Male        & 21{,}010 & 20{,}822 &  19 &  0.090 & 4.075 \\
$T_1$ & $\Delta$(F$-$M) & $-8{,}743$ & $-8{,}673$ & $-6$ & $+0.016$ & $+0.109$ \\
\midrule
$T_2$ & Female      & 12{,}141 & 12{,}079 &  17 &  0.140 & 4.179 \\
$T_2$ & Male        & 21{,}413 & 21{,}258 &  12 &  0.056 & 4.023 \\
$T_2$ & $\Delta$(F$-$M) & $-9{,}272$ & $-9{,}179$ & $+5$ & $+0.084$ & $+0.157$ \\
\bottomrule
\end{tabular}
\end{table}

\subsubsection*{Age composition}

Figure~\ref{fig:age-AFE-app} reports the $T_1 - C$ shift in age-band impression share (gender pooled), the algorithm-channel-only contrast on age composition: the algorithm channel reallocates working-age impressions toward women.

\begin{figure}[ht]
\centering
\includegraphics[width=0.78\textwidth]{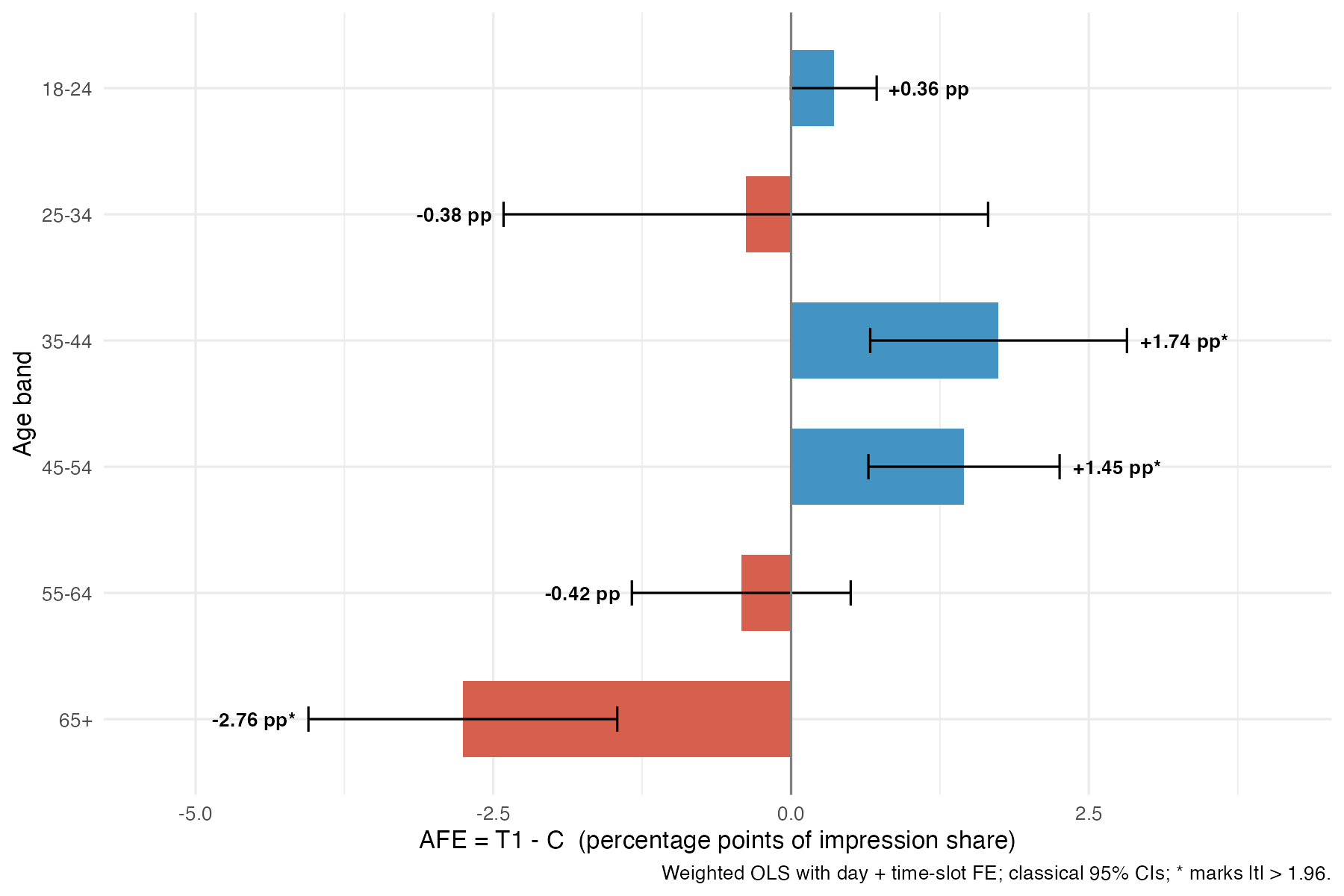}
\caption{$T_1 - C$ shift in age-band impression share, gender pooled (high bid). 95\% CIs; $*$ marks $|t|>1.96$.}
\label{fig:age-AFE-app}
\end{figure}

\subsection{Heterogeneity by bid level}
\label{app:bid-heterogeneity}

Replicating the decomposition at the low-bid experiment provides a second source of heterogeneity. Table~\ref{tab:bid-het} reports the decomposition of female impression share at each bid level.

\begin{table}[h]
\centering
\caption{Decomposition of female impression share by bid level.}
\label{tab:bid-het}
\begin{tabular}{lrr}
\toprule
Quantity (pp) & High bid: est (SE) & Low bid: est (SE) \\
\midrule
$\widehat{\mathrm{AFE}}$ on female share & $+2.27$ \,(0.77) & $+0.16$ \,(4.56) \\
$\widehat{\mathrm{ACE}}$ on female share & $-0.93$ \,(0.56) & $+0.12$ \,(7.94) \\
$\widehat{\mathrm{TE}}$  on female share & $+1.33$ \,(1.01) & $+0.29$ \,(7.01) \\
\bottomrule
\end{tabular}
\par\smallskip
\footnotesize\emph{Note.} Percentage points; standard errors in parentheses. Low-bid $n = 45$ cells.
\end{table}

The low-bid arms accumulated approximately $1{,}600$ impressions total, less than $5\%$ of the high-bid sample; the resulting standard errors exceed $4.5$~pp on every contrast and no low-bid coefficient is significantly distinguishable from zero. Only the precision of the decomposition differs across bid levels --- consistent with the simulation extension in Appendix~\ref{app:bid-transfer}, in which a lower bid reduces the precision of the impression-stage estimate, not its expected value. Formal inference is not possible on the low-bid arms given the sample size.

\subsection{Robustness across fixed-effects specifications}
\label{app:fe-robustness}

In a randomized experiment the arm coefficients are unbiased regardless of which fixed effects are included; the gain from adding controls is a reduction in residual variance and, correspondingly, tighter standard errors on the treatment coefficients. Table~\ref{tab:fe-robust} reports the aggregate decomposition under each combination of day and time-slot fixed effects.

\begin{table}[h]
\centering
\caption{Aggregate decomposition under alternative fixed-effects specifications, high-bid experiment.}
\label{tab:fe-robust}
\begin{tabular}{lrrr}
\toprule
Specification & $\widehat{\mathrm{AFE}}$ & $\widehat{\mathrm{ACE}}$ & $\widehat{\mathrm{TE}}$ \\
\midrule
(0) No FE                 & $+2.07$ \,(1.66) & $-0.68$ \,(1.90) & $+1.39$ \,(1.96) \\
(1) Day FE only           & $+2.06$ \,(0.85) & $-0.65$ \,(0.74) & $+1.42$ \,(0.93) \\
(2) Slot FE only          & $+2.11$ \,(1.13) & $-0.77$ \,(1.15) & $+1.34$ \,(1.24) \\
(3) Day + slot FE         & $+2.08$ \,(0.84) & $-0.68$ \,(0.70) & $+1.39$ \,(0.91) \\
\bottomrule
\end{tabular}
\par\smallskip
\footnotesize\emph{Note.} Percentage points; standard errors in parentheses. Aggregate-cell estimates; the primary estimates come from equation~\eqref{eq:main-spec}.
\end{table}

\subsection{Per-cell decomposition at the twelve (age $\times$ gender) cells}
\label{app:cell-decomp}

Table~\ref{tab:cell-decomp-app} disaggregates the decomposition of Table~\ref{tab:delivery-high} to the twelve individual (age $\times$ gender) cells. The finer decomposition is consistent with the age-band reading of Section~\ref{sec:empirical-age-gender} --- the algorithm channel's largest positive point estimates are in the 35--44 female ($+1.52$ pp) and 45--54 male ($+0.74$ pp) cells, with the largest negative estimate in the 65+ male cell ($-1.64$ pp) --- but the standard errors are wider than in the (working-age $\times$ gender) partition because each cell captures a smaller share of impressions.

\begin{table}[h]
\centering
\small
\caption{Per-cell decomposition of female impression share at (age $\times$ gender), high-bid experiment.}
\label{tab:cell-decomp-app}
\begin{tabular}{llrrr}
\toprule
Age & Gender & $\widehat{\mathrm{AFE}}$ & $\widehat{\mathrm{ACE}}$ & $\widehat{\mathrm{TE}}$ \\
\midrule
18--24 & female & $+0.13$ \,(0.07) & $-0.10$ \,(0.08) & $+0.04$ \,(0.09) \\
18--24 & male   & $+0.23$ \,(0.12) & $-0.21$ \,(0.12) & $+0.01$ \,(0.10) \\
25--34 & female & $+0.50$ \,(0.48) & $-0.87$ \,(0.37) & $-0.37$ \,(0.51) \\
25--34 & male   & $-0.88$ \,(1.21) & $-0.84$ \,(0.89) & $-1.72$ \,(1.27) \\
35--44 & female & $+1.52$ \,(0.53) & $-0.99$ \,(0.48) & $+0.53$ \,(0.58) \\
35--44 & male   & $+0.23$ \,(0.50) & $-0.94$ \,(0.45) & $-0.71$ \,(0.33) \\
45--54 & female & $+0.71$ \,(0.24) & $-0.43$ \,(0.19) & $+0.28$ \,(0.24) \\
45--54 & male   & $+0.74$ \,(0.31) & $-0.52$ \,(0.24) & $+0.23$ \,(0.32) \\
55--64 & female & $+0.33$ \,(0.24) & $+0.23$ \,(0.27) & $+0.57$ \,(0.24) \\
55--64 & male   & $-0.75$ \,(0.32) & $+1.16$ \,(0.34) & $+0.41$ \,(0.36) \\
65+    & female & $-1.11$ \,(0.34) & $+1.48$ \,(0.38) & $+0.36$ \,(0.41) \\
65+    & male   & $-1.64$ \,(0.45) & $+2.04$ \,(0.57) & $+0.40$ \,(0.61) \\
\bottomrule
\end{tabular}
\par\smallskip
\footnotesize\emph{Note.} Percentage points; standard errors in parentheses. Separate weighted OLS per cell with day and time-slot fixed effects; HC1 SEs.
\end{table}

\section{Further details on the methodology}
\label{app:methodology-details}

This appendix collects supporting material for the methodology developed in
Section~\ref{sec:methodology}. Section~\ref{app:no-regression}
explains why a naive regression adjustment for $S$ does not recover the direct
effect. Section~\ref{app:threats} discusses two threats to design
validity---interference and power. Section~\ref{app:cde-vs-nde} distinguishes
the population-average decomposition our design identifies from the unit-level
natural direct effect. Section~\ref{app:inference-details} lays out the
inference machinery: cluster-robust standard errors, the wild cluster
bootstrap, and multiple-testing adjustment via Romano-Wolf.

\subsection{Why a regression adjustment for $S$ is not a fix}
\label{app:no-regression}

A common but biased remedy for the confounding described in
Section~\ref{sec:methodology} is to include $S$ as a control variable
in a regression of $Y$ on the treatment indicator:
\begin{equation}
Y_i \;=\; \alpha + \beta D_i + \gamma S_i + \varepsilon_i.
\label{eq:bk-app}
\end{equation}
This is the structure underlying \allowbreak{the classic Baron--Kenny regression approach}; the bias
of the approach is \allowbreak{detailed in \citet{vanderweele2015explanation}}. Conditioning on a post-treatment variable
that is itself caused by the treatment opens a non-causal path
between $D$ and $\varepsilon$, so $\hat{\beta}$ in~\eqref{eq:bk-app}
does not estimate the direct effect even under random assignment of
$D$. The unbiasedness conferred by randomization is broken the moment
a post-treatment regressor is introduced. The only assumption-light
remedy is to make $S$ itself respond to a randomization that is
independent of treatment assignment---in other words, to redesign the
experiment.

\subsection{Threats to design validity}
\label{app:threats}

Two issues warrant attention at the design stage: interference between
subjects and statistical power under the three-way sample split.

\paragraph{Spillovers and SUTVA.}
If subjects communicate, or if $S$ has equilibrium feedback across
units, the three-arm design must be combined with cluster-level
randomization or exposure-mapping analysis
\citep{aronow2017estimating}. The decomposition of
Section~\ref{sec:identification} continues to hold within exposure
conditions, but the conditioning set must be made explicit. In the
ad-delivery setting the most plausible spillover channel is algorithmic
learning across arms: if the platform's serving rule updates from
engagement feedback observed in one arm, that update can leak into
another arm's realized $S$. In the field experiment we limit this concern by restricting the
sample to the campaign's initial week, within which engagement feedback
has little scope to update the serving rule; the Assumption-(A3) check
of Section~\ref{sec:empirical-delivery} does not reject.

\paragraph{Power.}
Splitting the sample three ways reduces the precision of any single
pairwise contrast by roughly a factor of $\sqrt{2/3}$ compared with a
two-arm design at the same total sample size. Because $T_1$ enters
both $\widehat{\mathrm{ACE}}$ and $\widehat{\mathrm{AFE}}$, an unequal
allocation favoring $T_1$ minimizes the total variance of the
decomposition: with equal outcome variances across arms, the sum
$\Var(\widehat{\mathrm{AFE}}) + \Var(\widehat{\mathrm{ACE}}) =
\sigma^2(1/n_C + 2/n_{T_1} + 1/n_{T_2})$ is minimized at the
allocation $1:\sqrt{2}:1$. In practice the choice of allocation
should also weigh the relative importance of estimating the algorithm
channel and the creative channel and the anticipated effect sizes on
each. These considerations compound the already unfavorable power
economics of advertising experiments \citep{lewis2015unfavorable}.

\subsection{Controlled vs.\ natural direct effect}
\label{app:cde-vs-nde}

A subtlety in the interpretation of the AFE and ACE merits emphasis.
The identification result of Section~\ref{sec:identification} uses the
\emph{marginal} distribution of $S(B)$ when imposing $S$ in arm $T_1$,
not each subject's own counterfactual $S_i(B)$. In the language of
\citet{pearl2001direct} and \citet{vanderweele2015explanation}, the
AFE and ACE are therefore \emph{population-average} effects --- controlled
direct effects marginalized over the distribution of $S(B)$ --- rather
than unit-level natural effects.

The unit-level analogue of the ACE,
\[
\E\bigl[Y_i(B, S_i(B)) - Y_i(A, S_i(B))\bigr],
\]
would additionally require the cross-world independence assumption
$Y(a, s) \indep S(a') \mid X$ for $a \neq a'$
\allowbreak{\citep{vanderweele2015explanation}}.

For most applied purposes the population-average decomposition is the
policy-relevant object: it answers the counterfactual ``what would the
contrast $A - B$ look like if delivery were equalized at $B$-levels?''
Researchers who require the unit-level decomposition must either
invoke sequential ignorability or use a
sensitivity analysis along the lines of \allowbreak{\citet[chs.~2--3]{vanderweele2015explanation}}
to bound the resulting bias.

\subsection{Inference: cluster-robust standard errors, wild cluster bootstrap,
and Romano-Wolf adjustment}
\label{app:inference-details}

The estimator formula~\eqref{eq:varhat} in Section~\ref{sec:methodology}
is the population-level Neyman variance for a difference of arm means.
It is valid when impressions within an arm are independent and
homoskedastic. In production ad-delivery experiments these assumptions
are typically violated on two counts: the algorithm's daily state
induces within-arm serial structure, and the outcome variance is
heteroskedastic across (age $\times$ gender) cells and across
day-of-week.

We therefore use cluster-robust (CR2) standard errors, computed on
the impression-weighted regression at the disaggregated (arm $\times$ day
$\times$ time-slot $\times$ age $\times$ gender) level, clustered at the
(day $\times$ time-slot) level ($G = 19$) --- the level at which
delivery-auction shocks are shared. The HC1 estimator on the
aggregate-cell regression is reported as a robustness specification in
Appendix~\ref{app:further-empirical}.

\paragraph{Multiple testing.}
Multiple-testing across the three pairwise contrasts (and, in the
disaggregated analyses of Section~\ref{sec:empirical-age-gender},
across the $36$ cell-level contrasts) is controlled by the
Romano-Wolf step-down procedure, implemented via a wild cluster
bootstrap with (day $\times$ time-slot)-level Rademacher weights ($B = 5{,}000$) and
centered $t$-statistics. Implementation details are given in the
replication code.

\end{document}